\documentclass[acmsmall, screen]{acmart}

\usepackage{siunitx}
\usepackage{algorithmicx}
\usepackage{algpseudocode}
\usepackage{caption}
\usepackage{booktabs}
\usepackage{mdframed}
\usepackage{hyperref}

\setcopyright{none}
\renewcommand\footnotetextcopyrightpermission[1]{}
\begin{document}

\title{Starfield: Demand-Aware Satellite Topology Design for Low-Earth Orbit Mega Constellations}

\author{Shayan Hamidi Dehshali}
\email{hamididehshali.1@buckeyemail.osu.edu}
\author{Tzu-Hsuan Liao}
\email{liao.834@buckeyemail.osu.edu}
\author{Shaileshh Bojja Venkatakrishnan}
\email{bojjavenkatakrishnan.2@osu.edu}
\affiliation{
  \institution{The Ohio State University}
  \city{Columbus}
  \state{Ohio}
  \country{USA}
}

\begin{abstract}
Low-Earth orbit (LEO) mega-constellations have gained popularity as high-coverage, high-capacity networks, serving as a potential backbone for next-generation Internet. With the deployment of laser communication terminals on commercial LEO satellites, high-bandwidth and low-latency inter-satellite routing has become feasible. However, the limited number of laser terminals per satellite, coupled with the slow acquisition process and instability of laser links, makes forming a stable satellite topology particularly challenging. Although several topology patterns—such as +Grid and Motif—have been proposed, none of these approaches account for regional traffic demand, ground station locations, or the spatial properties of the satellite constellation. Given the sparse clustering of population on Earth and the isolation of rural areas, traffic patterns are inherently non-uniform, providing an opportunity to orient inter-satellite links (ISLs) according to these traffic patterns.

In this paper, we propose Starfield, a novel inter-satellite topology design based on a demand-aware geometrical metric. We first formulate a vector field on the constellation’s shell according to traffic flows and define a corresponding Riemannian metric on the spherical manifold of the shell. This metric, combined with the spatial geometry, is used to assign a distance to each potential ISL, which we then aggregate over all demand flows to generate a distance for each satellite’s link selection. Inspired by +Grid, each satellite then selects the link with the minimum Riemannian distance along with its corresponding angular links. To evaluate Starfield, we developed a custom, link-aware, and link-configurable packet-level simulator, comparing it against LAMP, Motif, +Grid, and Random topologies. For the Phase 1 Starlink constellation, simulation results show up to a 30\% reduction in hop count and a 15\% improvement in stretch factor across multiple traffic distributions. Moreover, static Starfield, an inter-orbital link matching modification of Starfield, achieves a 15\% improvement in stretch factor under realistic traffic patterns compared to LAMP. Experimental results further demonstrate that Starfield is robust to fluctuations in traffic demand. Finally, we provide a theoretical analysis establishing a lower bound and estimating the algorithm’s performance.
\end{abstract}



\maketitle

\section{Introduction}
Low-earth orbit (LEO) satellite networks are rapidly emerging as a promising modality for providing Internet access to areas where terrestrial wired and wireless networks are unavailable, such as remote locations or areas that have suffered a natural or man-made calamity~\cite{ma2023network,hurricanehelene,starlinkrussiaukraine,starlinkmaritime,starlinkemergency}. Multiple providers---SpaceX's Starlink~\cite{starlink}, Eutelsat's OneWeb~\cite{oneweb}, Amazon's Leo~\cite{leo}, Iridium~\cite{iridium}---have already deployed LEO satellite constellations comprising thousands of satellites and are steadily growing their service regions worldwide~\cite{starlinkavailability}. The largest of these is the Starlink constellation which comprises of more than 7,800 operational satellites in orbit serving $>$ 2.7 million customers, with plans for adding 42,000 more in the near future~\cite{starlinknetupdate,starlinknumfuture,mohan2024multifaceted}. Starlink satellites are equipped with 200 Gbps high-speed optical inter-satellite links (ISLs) and 20 Gbps phased array radio antennas for satellite-to-ground communication~\cite{starlinkstech, starlinkbandwidthusa}. With the help of a small, portable radio terminal (called ``Dishy'') users on the ground can obtain broadband Internet access, suitable for Web browsing, video streaming, gaming, and teleconferencing, etc., with a median latency of 33 ms~\cite{starlinklatency}. While high-bandwidth rural coverage was a primary motivation behind the development of Starlink, the rapid growth in the number and capabilities of low Earth orbit (LEO) satellites has created new opportunities for integrating space and terrestrial cellular networks as well. This integration promises significant improvements in network throughput and latency and positions satellite constellations as a key component of future 6G systems~\cite{xiao2022leo, andrews20246}.

A key feature of Starlink satellites is the high-speed optical ISLs, which allow packets to be routed directly from one satellite to another nearby satellite. While early designs of Starlink did not use direct satellite-to-satellite routing, favoring a simpler ground-to-satellite-to-ground routing (also called as the "bent-pipe" routing~\cite{abubakar2024choosing}), direct inter-satellite routing is more latency and bandwidth efficient particularly for long-haul traffic, and is being implemented to supplant bent-pipe routing~\cite{chen20243, hauri2020internet}. Despite the intuitive benefits of ISL routing, ISLs are not without limitations: (1) each Starlink satellite is equipped with only 3 to 5 ISL transceivers; (2) the maximum range of an ISL is limited to 5000 km, with bandwidth decaying with link distance~\cite{chaudhry2021laser, shang2025channel}; (3) establishing a fresh ISL between two satellites can take several seconds due to pointing, tracking and acquisition delays~\cite{carrizo2020optical, brashears2024achieving}; (4) the high relative-velocity between satellites, and the resulting rapid inter-orbital distance variation, makes maintaining persistent ISLs extremely challenging.

On the other hand, the choice of ISLs induces a network topology of capacitated links that directly determines packet round-trip times, path bandwidth and ultimately, end-user quality-of-service~\cite{chen20243}. A popular proposal for ISL construction is the +Grid topology~\cite{wood2001internetworking}, in which each satellite forms ISLs with four nearby satellites (two satellites in the same orbit, and one satellite in each of the adjacent orbits) in an approximate grid-like pattern. This is a simple yet effective construction as the grid structure ensures all four ISLs are utilized in all the satellites, forms ISLs that are persistent, uniformly covers all directions, and provides a number of path choices between any source and destination regardless of their locations. 

A fundamental shortcoming of the grid topology is its agnosticism to network traffic demand patterns, which are inherently non-uniform due to the sparse clustering of populations and the isolation of rural areas on Earth. The grid naturally achieves near-optimal path lengths for traffic flows that are parallel to its two principal axes. For flows along directions that are skewed with respect to the axes, the shortest path on the grid zig-zags leading to an inflated round-trip time. In the worst case this inflation, or {\em stretch factor}, can be more than a factor two larger than the geodesic between the source and the destination. Recent work has generalized the grid's plus $(+)$ pattern of neighbor selection to other motifs (e.g., $\times$ etc.)~\cite{bhattacherjee2019network}. However, even when the motifs are carefully chosen to best fit a known traffic demand pattern, the notion of using the same motif on all satellites regardless of ground stations' locations can again lead to suboptimal stretch factor for many demands. 

In this work, we present Starfield, a {\em demand-aware} LEO satellite topology construction algorithm that aligns ISLs in the direction of the demand flows to minimize stretch factor. Starfield considers the satellites as embedded within a 2-dimensional spherical manifold with a general (Riemannian) metric.\footnote{The metric provides a notion of distance between points on the spherical manifold.} 
The metric is modified such that points along the direction of a demand flow are close in distance and thus provides a convenient method for satellites to prioritize nearby satellites for link selection. In a way, Starfield is a generalization of the +Grid topology. While a +Grid satellite connects with its four closest neighbors with closeness measured per the spherical Euclidean metric, Starfield encourages satellites to connect to close neighbors with closeness measured via a demand-aware Riemannian metric. 

To compute the metric, Starfield constructs a ``demand field''---a vector field on the spherical manifold---for each traffic flow, where the field is the strongest along the geodesic from the source to the destination, and decays at points farther away from endpoints. The direction and value of the demand field indicates our desired propensity for packets to flow in that direction. The idea of demand fields is loosely motivated by the concept of fields between dipoles in electromagnetism~\cite{griffiths2023introduction}; though, the choice of field function we use is different and better suited to our setting. The link selection is made at each satellite with aggregated distance of the metrics corresponding to each demand field: we set the metric to be small in directions along the field lines and large in directions orthogonal to the field lines. Demand fields naturally help with demand aggregation: If multiple small demand flows exist along a particular direction in an area, the influence of each corresponding demand field—and, consequently, the metric—accumulates, increasing the overall importance of establishing ISLs in that direction. Demand fields also enhance fault tolerance and mitigate demand interference: if some ISLs along the geodesic between a source and destination are offline or misaligned, packets can still follow nearby, near-optimal routes. Unlike the one-size-fits-all approach of +Grid, Starfield produces a topology that is finely tuned to the orientations of all traffic demand flows, without being biased toward flows along any specific principal directions.

We evaluate Starfield using a custom packet-level simulator that is faster, lighter, and more link-aware and link-configurable than existing state-of-the-art simulators and emulators, particularly by enabling flexible topology configuration and Shannon–Hartley data-rate modeling~\cite{shannon1948mathematical}. Across diverse traffic distributions, our results demonstrate up to a 15\% improvement in stretch factor and a 30\% reduction in hop count compared to state-of-the-art method LAMP, with static Starfield, an inter-orbital link matching modification of Starfield, achieving a 15\% improvement in stretch factor. Furthermore, we show that Starfield remains robust under traffic perturbations, exhibiting less than 3\% stretch factor degradation under Gaussian noise.

We present a theoretical analysis of Starfield on a simplified two-dimensional flat manifold, showing that the shortest path length exceeds the geodesic flow by a factor determined by the angle between the demand field and the geodesic. This analysis explains classical results such as the $\sqrt{2}$ stretch between diagonally opposite points on a square grid and provides a lower bound on the worst-case stretch for arbitrary traffic demands. The contributions of this paper are as follows: 
\begin{itemize}
\item We propose a demand-aware algorithm for inter-satellite link (ISL) topology design, based on a Riemannian metric derived from a demand vector field.
\item We develop a custom, link-aware, and link-configurable packet-level simulator, which we use to evaluate Starfield both quantitatively and visually.
\item We present experimental results demonstrating the superior performance of Starfield in network metrics—such as stretch factor, hop count, and round-trip time—compared to the +Grid, LAMP, Motif, and Random baselines.
\item We present a mathematical analysis of Starfield on a simplified 2-dimensional flat manifold, providing an estimate of the shortest path and a lower bound on the worst-case stretch.
\end{itemize}

\section{Background}\label{background}
\subsection{Satellite Constellations}
The basic component of the constellation is a circular orbit containing satellites that are evenly spaced and revolve around the Earth at a constant velocity in a common direction. A \textit{shell} is defined as a collection of such orbits operating at the same altitude and corresponding orbital velocity. Each shell consists of $N_O$ orbits, with $N_S$ satellites per orbit, for a total of $n = N_SN_O$ satellites. To achieve a wide coverage area and avoid collisions, the orbits are carefully arranged in a staggered manner. Following convention, we use the earth-centered inertial frame (ECI) of reference~\cite{earthcenteredinertial}, in which the origin lies at the center of the earth, the $Z$-axis passes through the north pole, the $X$ axis is fixed in a direction relative to the celestial sphere, e.g., Prime Meridian at a specific time, and the $Y$ axis perpendicular to the $X$ and $Z$ axes. Note that the $XY$ plane coincides with the equatorial plane. 

The position of the $j$-th orbit, for $j=1,2,\ldots,N_O$, is determined through three steps. First, the orbit is laid flat on the equatorial ($XY$) plane. Second, the orbit is rotated $i^\circ$, \textit{orbital inclination}, clockwise about the $X$-axis. Third, the orbit is rotated about the $Z$-axis by an angle of $\frac{2\pi j}{N_O}$ clockwise. This latter rotation, referred to as the \textit{orbital phase}, is intended to uniformly cover the corresponding latitude range on Earth. Phase 1 Starlink consists of a single shell with $N_O = 72$ orbits with $N_S = 22$ satellites per orbit, deployed at an inclination of $i= 53^\circ$~\cite{starlinksatparameters}. 

\subsection{Communication Links}

\noindent 
{\bf User terminal.} A Starlink end-user communicates with a satellite through a small user terminal radio, the size of a pizza box. The user terminal has a phased array antenna and tracks multiple satellites at any time, that are within its $\sim 100^\circ$ field of view of the sky. The terminal chooses the best satellite available for transferring data; it is possible for the terminal to switch satellites up to 10 times per second~\cite{starlinkbeamswitching}. Communication occurs over the $K_u$ band for both uplink and downlink, with a promised maximum of 35 Mbps upload and 200 Mbps download speed~\cite{starlinkspecs}. 

\smallskip 
\noindent 
{\bf Satellite.} A starlink satellite is equipped with multiple phased array antennas in the $K_u, K_a$ and $E$ bands for communication with user terminals and ground stations (discussed next)~\cite{iyer2023system}. Operating on a 2 GHz band divided into eight 250 MHz channels over multiple beams, a single satellite is capable of serving a few thousand users simultaneously~\cite{iyer2023system}. A satellite can also communicate to a nearby satellite via an optical inter-satellite laser link with a data rate of over 200 Gbps. Each satellite is equipped with 3, 4 or 5 ISLs~\cite{starlinkstech, chaudhry2021laser}, that it uses to establish links with nearby satellites within the same or adjacent orbital planes. Due to the high relative velocities, it is difficult to establish ISLs between nearby satellites whose orbits cross each other (i.e., have a large difference in phase close to $180^\circ$). While intra-orbit ISLs are permanent, it is possible for inter-orbit ISLs to temporarily disconnect due to the distance going out of range (beyond $5000$ km). The setup delay of establishing an ISL between two satellites ranges from a few seconds to 10s of seconds due to pointing, acquisition, and tracking complexities~\cite{kaushal2017acquisition, carrizo2020optical}, but is projected to improve to the millisecond time scale in the near future~\cite{bhattacharjee2023laser}.  

\smallskip 
\noindent 
{\bf Ground station.} The Starlink system also encompasses more than 100 ground stations, which are large facilities that act as gateways between the satellite network and the Internet~\cite{starlinknetupdate}. A ground station houses multiple large dishes (and/or smaller phased arrays) that communicate with satellites on the $K_a$ and $E$ bands at few 10s of Gbps. The capacity is expected to increase to Tbps with advancements in $E$-band antenna design.

\subsection{Riemannian Manifolds and Metric}\label{background:riemannian}
We provide an informal overview of Riemannian manifold and metric, and we refer to Do Carmo~\cite{do1992riemannian} for a formal treatment of the subject. Riemannian geometry is a generalization of Euclidean geometry to ``curved'' spaces, where the metric (the yardstick) for measuring distance itself can vary depending on {\em where} the measurement is being done. Manifolds---smooth surfaces without abrupt boundaries or breaks---are the focal objects of study in Riemannian geometry. In this paper, we focus on the two-dimensional spherical surface upon which the orbits lie as the manifold of interest.\footnote{Despite being a three-dimensional object, the spherical surface is considered a two-dimensional manifold since the local neighborhood around each point on the sphere ``looks like'' a two-dimensional Euclidean space.} A fundamental principle in Riemannian geometry is to consider the manifold as a self-contained object that is not embedded within a larger parent space. Instead, the manifold is studied using an ``atlas'' which is a collection of smooth, injective maps from open sets of $\mathbb{R}^2$ to the manifold that together cover the manifold. The key geometric properties of a manifold (e.g., curvature, distance, etc.) do not depend on the particular choice of atlas used to study the manifold.  

Consider any curve $\gamma: [0, 1] \rightarrow M$ on the manifold. 
The curve may be viewed as the path of a particle moving from $\gamma(0)$ to $\gamma(1)$ over $\lambda \in [0,1]$. The length of the curve $\gamma$, or equivalently the distance traveled by the particle can be given as: 
\begin{align}
\int_{\lambda=0}^1 \left\lVert  \frac{d\gamma(\lambda)}{d\lambda} \right\rVert d\lambda = \int_{\lambda=0}^1 \sqrt{ \langle \frac{d\gamma(\lambda)}{d\lambda}, \frac{d\gamma(\lambda)}{d\lambda} \rangle } \, d\lambda  
\end{align}
where $\frac{d\gamma(\lambda)}{d\lambda}$ is the velocity tangent of the particle at $\lambda$ and $\langle \cdot, \cdot \rangle$ denotes the inner product. Since the velocity tangent may, in general, point “outside” the manifold and there is no ambient space beyond the manifold, we define the notions of tangents and inner products using the atlas. 

Let $p = \gamma(\lambda)$ be the point on the curve $\gamma$ at $\lambda$. Consider an atlas for $M$. Let $c: U \rightarrow M$ be a map from the atlas, mapping an open set $U \subseteq \mathbb{R}^2$ to $M$ such that $p \in c(U)$. A tangent at $p$ is any vector originating from $c^{-1}(p)$ in $\mathbb{R}^2$. A Riemannian metric at $p$ is any positive-definite matrix $ g_p \in S_2^{++}$.\footnote{As with the maps in the atlas, we require $g_p$ to vary smoothly over the mainfold.}
For any two tangents $\tau_1, \tau_2 \in \mathbb{R}^2$ at $p$, the inner product between the tangents is defined as $\langle \tau_1, \tau_2 \rangle = \tau_1^T g_p \tau_2.$ Therefore, the norm of a tangent $\tau \in \mathbb{R}^2$ of $p$ is $ \| \tau \| = \sqrt{\tau^T g_p \tau}$. The velocity tangent for the curve $\gamma$ at $\lambda$ is the vector $\frac{d c^{-1}(\gamma(\lambda))}{d\lambda}$. The length of the curve $\gamma$ can thus be given as:
\begin{align}
\int_{\lambda=0}^1 \left\lVert  \frac{d c^{-1}(\gamma(\lambda))}{d\lambda} \right\rVert d\lambda = 
\int_{\lambda=0}^1 \sqrt{ \frac{d c^{-1}(\gamma(\lambda))}{d\lambda} ^T g_{\gamma(\lambda)} \frac{d c^{-1}(\gamma(\lambda))}{d\lambda}} d\lambda. 
\end{align}

\section{System Model and Objective}\label{modelObjective}
We consider a LEO satellite system comprising of $N_O$ orbits with $N_S$ satellites per orbit. $S = \{s_1, s_2, \ldots, s_{n}\}$ denotes the set of all satellites. Each satellite has $\kappa$ ISLs available for forming connections with other satellites that lie within its communication range. We assume an altitude of $h$ for the satellite shell, and a maximum ISL range of $r$. $R^t$ denotes a graph with node set $S$, where an edge exists between $s_i$ and $s_j$ if they are within communication range at time $t$. For a potential ISL $(s_i, s_j) \in R^t$, the packet propagation latency along the link is given by $\frac{\|\vec{P^t_{s_i}} - \vec{P^t_{s_j}}\|}{c_\mathrm{light}}$ where $\vec{P^t_{s_i}}$ denotes the position vector of $s_i$ at $t$, and $c_\mathrm{light}$ is the speed of light in vacuum. The capacity of an ISL $(s_i, s_j) \in R$ is given by $c_1 \log(1 + \frac{c_2}{\|\vec{P^t_{s_i}} - \vec{P^t_{s_j}}\|^2})$ where $c_1$, and $c_2$ are noise parameters. Each satellite has a per-ISL buffer size of $B$ packets; we assume all packets are of the same size of $b$ bytes. Packets queued in the buffer are transmitted in a first-in first-out order. Apart from queuing loss, transmitted packets are not corrupted or otherwise lost during transit. 

$U = \{ u_1, u_2, \ldots, u_{m} \}$ represents the set of ground stations, which are fixed on the Earth’s surface and serve as both sources and destinations of data packets. In our model, we do not distinguish between user terminals and ground stations. Ground stations exchange data according to a known\footnote{Prior works demonstrate techniques for measuring demand that account for uneven population distribution and rural remoteness on Earth \cite{li2025small, misa2025sustainability}.} 
demand pattern $\Delta^t \in \mathbf{R}^{m \times m}$, where the $(i,j)$-th entry $\Delta^t_{ij}$ denotes the traffic size required between ground stations $u_i$ and $u_j$ at $t$. At $u_i$, packets destined for $u_j$ are generated according to a Poisson process with rate $\Delta^t_{ij}$ and transmitted in respective order. Similar to ISLs, the propagation delay between a $u_i$ to a $s_j$ (in either direction) within range—defined by a minimum elevation angle of $25^\circ$ above the horizon— is $\frac{\|\vec{P^t_{u_i}} - \vec{P^t_{s_j}}\|}{c_\mathrm{light}}$, where $\vec{P^t_{u_i}}$ denotes the position vector of $u_i$ at $t$. As with ISLs, we use the Shannon formula to model the uplink and downlink bandwidth of the ground-satellite links. To forward a packet, a ground station computes the shortest path to the destination over the topology graph $G$---without intermediate ground-station relays---with link costs weighted by propagation delays, and sends the packet to the satellite at the first hop. Our model is restricted to the IP layer, ignoring effects from other OSI layers, and also omits packet processing delays at the satellites.

\subsection{Objective}
For each flow $(u_i,u_j)$, we define a stretch $\sigma^T \geq 1$ over period $T$ as follows. Let $d_{\mathbf{S}^2}(u_i,u_j)$ denote the geodesic distance between $u_i$ and $u_j$ on Earth ($\mathbf{S}^2$). We use $d_{\mathbf{S}^2}(u_i,u_j)$ as a theoretical lower bound on the packet propagation length. Let $d^{\frac{T}{2}}_G(u_i, u_j)$ be the achieved packet propagation length from $u_i$ to $u_j$ over the topology $G$ at $\frac{T}{2}$\footnote{Because the satellites are uniformly spaced, the constellation pattern repeats every few minutes, allowing us to use a central position for short periods.}
, routed through space. We define the weighted sum of stretch over period $T$ as:
\begin{align}
\sigma^T = \sum_{(u_i,u_j) \in \Delta^T} \Delta^T(u_i, u_j)\frac{d^{\frac{T}{2}}_G(u_i, u_j)}{d_{\mathbf{S}^2}(u_i,u_j)}
\end{align}
where $\Delta^T$ is $\frac{1}{T}\sum_{t \in T} \Delta^t$. We formulate the problem of designing a topology $G$---subject to the degree and range constraints described above---as the minimization of the $\sigma^T$.

\section{Motivation}\label{motivation}

\subsection{Demand-Aware Topology Design}\label{motivation:demandAware}

\begin{figure}[!t]
  \centering
  \begin{minipage}[t]{0.48\linewidth}
    \centering
    \includegraphics[width=\linewidth]{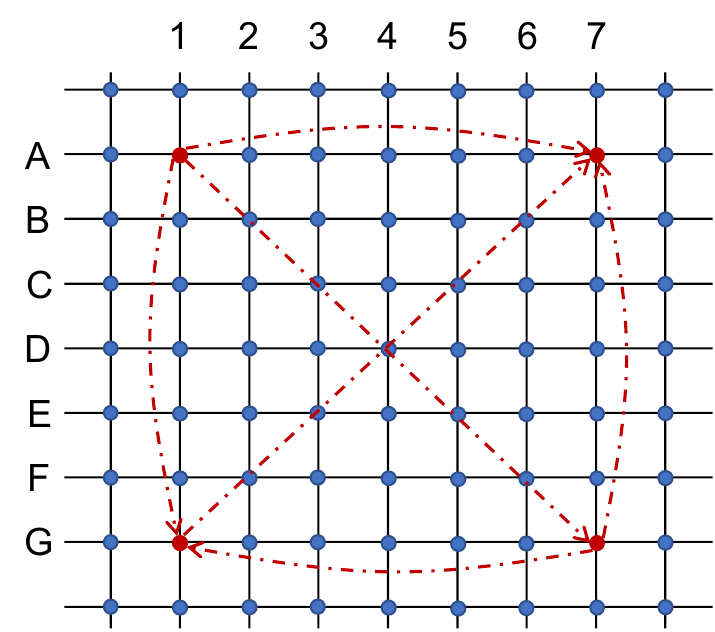}
  \end{minipage}\hfill
  \begin{minipage}[t]{0.48\linewidth}
    \centering
    \includegraphics[width=\linewidth]{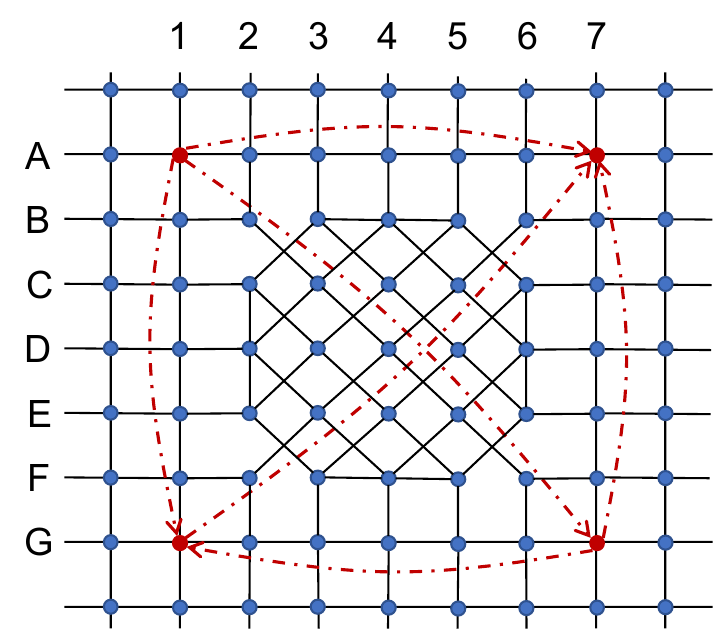}
  \end{minipage}
  \caption{+Grid topology (left) and diagonally oriented topology (right) on a grid of satellites.}
  \label{motivation:demandAware:grid}
\end{figure}

We illustrate the advantages of demand-aware network topology design through a stylized example. For simplicity, we consider a uniform arrangement of satellites forming a grid on a two-dimensional flat plane, as shown in Figure~\ref{motivation:demandAware:grid}. The distance between adjacent satellites on the grid is one unit, and each satellite can establish connections with up to four neighboring satellites located within a distance of two units. Rows are labeled alphabetically and columns numerically; for instance, $A1$ denotes the satellite located in row $A$ and column $1$. We assume traffic demands among the satellites $A1$, $A7$, $G1$, and $G7$, as depicted.\footnote{For simplicity, we assume traffic demands directly between satellites rather than between ground stations.}

Figure~\ref{motivation:demandAware:grid} (left) shows the satellites connected using a +Grid topology. In this configuration, the demands $(A1 \rightarrow A7)$, $(A1 \rightarrow G1)$, $(G7 \rightarrow A7)$, and $(G7 \rightarrow G1)$ are routed along paths that match the straight-line (geodesic) connections between their respective sources and destinations, yielding a stretch factor of 1. However, for the cross demands $(A1 \rightarrow G7)$ and $(G1 \rightarrow A7)$, the routing paths are not aligned with their corresponding geodesics. For example, for the demand $(A1 \rightarrow G7)$, the shortest available path has a length of 12 units, whereas the geodesic distance between $A1$ and $G7$ is $6\sqrt{2}$ units, resulting in a stretch factor of $\sqrt{2} \approx 1.41$.

Next, we consider the diagonally oriented topology shown in Figure~\ref{motivation:demandAware:grid} (right). In this configuration, ISLs are oriented to align as closely as possible with the dominant demand flow directions. As in the +Grid topology, the demands $(A1 \rightarrow A7)$, $(A1 \rightarrow G1)$, $(G7 \rightarrow A7)$, and $(G7 \rightarrow G1)$ achieve a stretch factor of 1. Importantly, the cross demands also experience improved performance, as the ISLs along their geodesic paths are now oriented parallel to the demand directions. Specifically, for the demand $(A1 \rightarrow G7)$, the shortest path has a length of $4\sqrt{2} + 4$, yielding a stretch factor of approximately $1.13$.\footnote{This value approaches 1 as the grid size increases.}

\subsection{Geometry of Real-World Demand Patterns} \label{motivation:demandGeometry}
\begin{figure}[!t]
\centering
  \begin{minipage}[t]{0.48\linewidth}
    \centering
    \includegraphics[width=\linewidth]{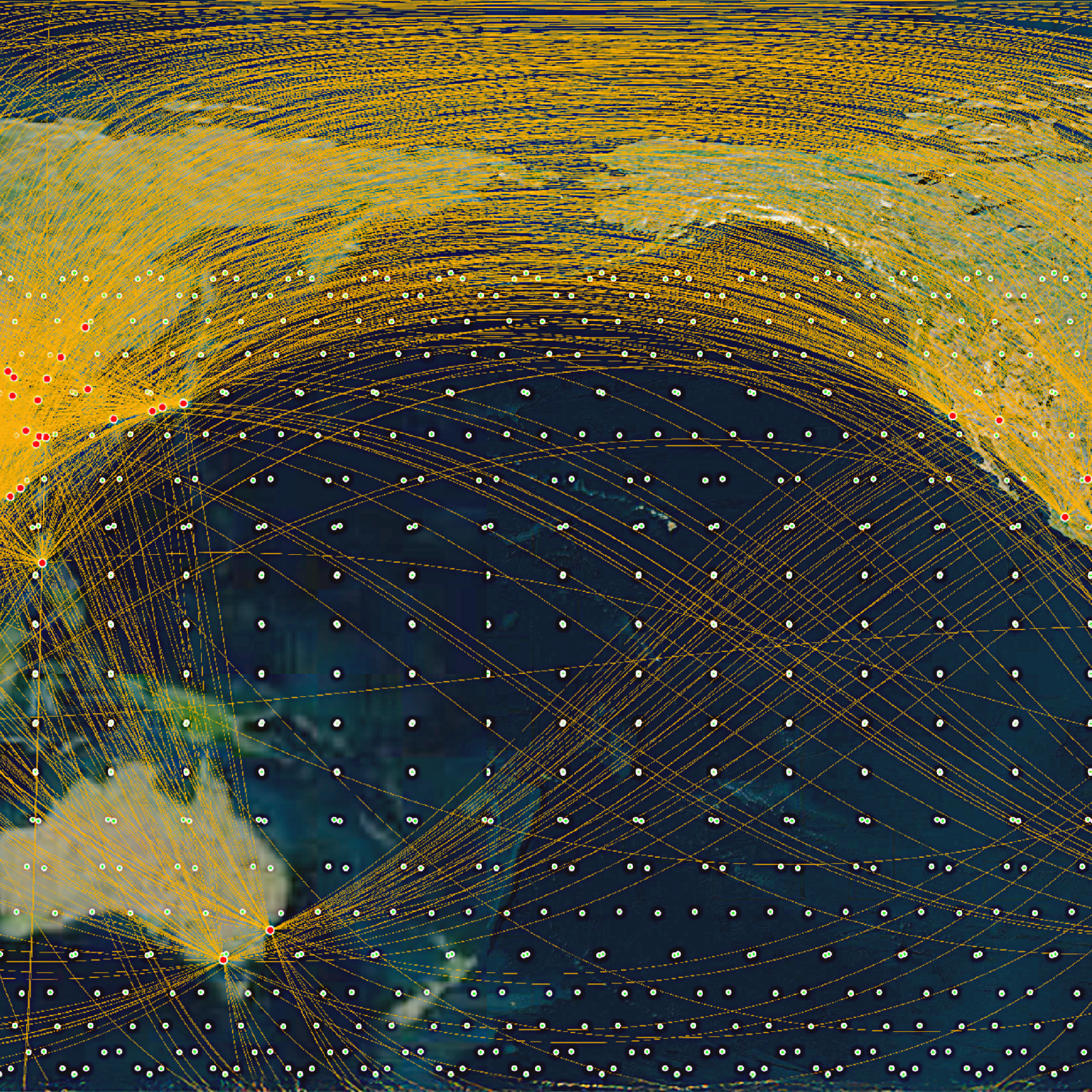}
  \end{minipage}\hfill
  \begin{minipage}[t]{0.48\linewidth}
    \centering
    \includegraphics[width=\linewidth]{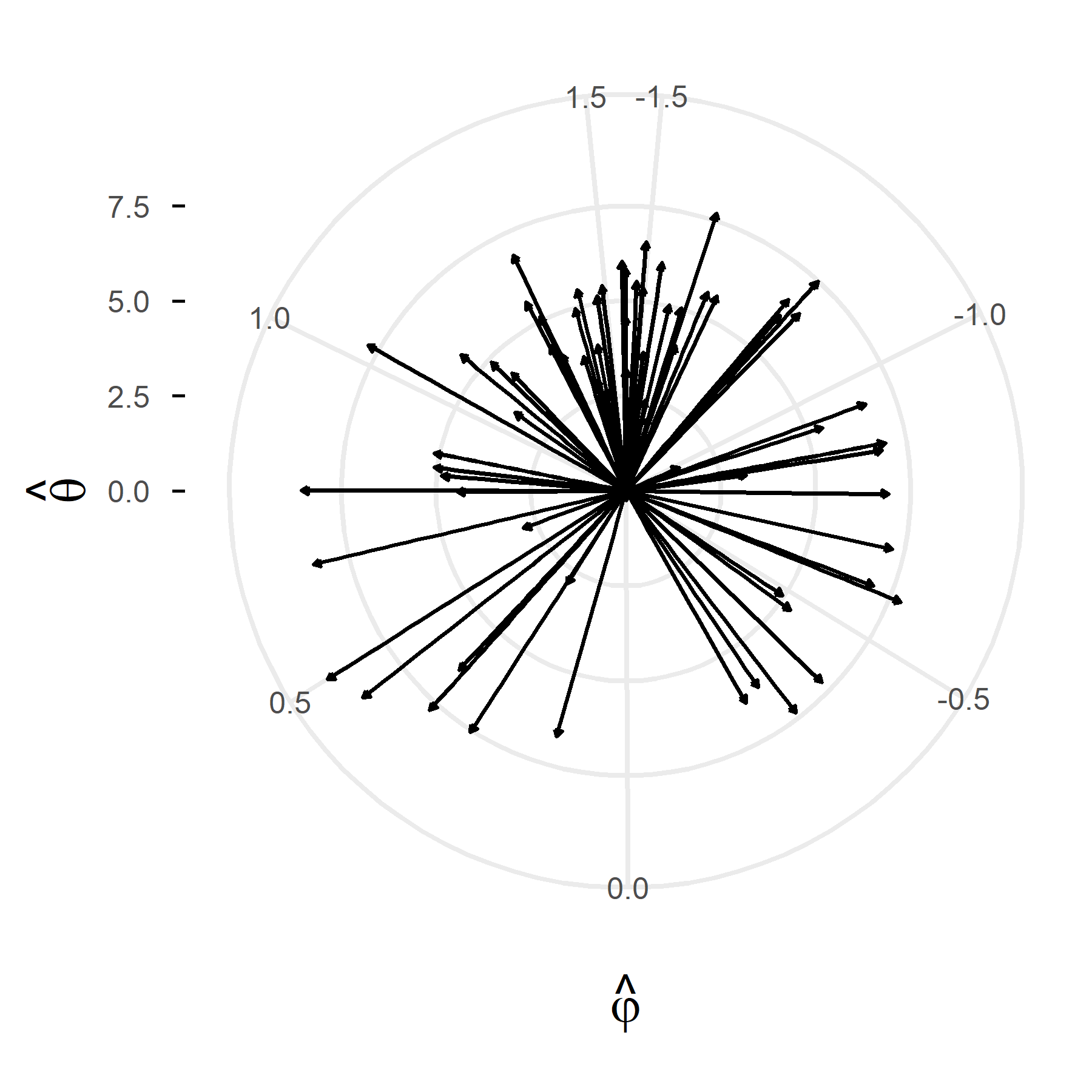}
  \end{minipage}
  \caption{Geodesic flows (orange lines) between 100 highly populated cities under the distance–population demand pattern (left), with line thickness representing traffic volume. Green and red dots denote Phase 1 Starlink satellites and ground stations, respectively. The corresponding regional directional flow components are shown as a log-scaled arrow plot with weighted mean resultant length 0.72 (right).}
  \label{motivation:demandGeometry:pattern}
\end{figure}
A demand-aware topology design is unlikely to provide significant benefits over +Grid if traffic demands are uniformly distributed in all directions at all locations. A tangible advantage arises when, within any local region, there exists one or two principal directions along which the majority of demands flow. Empirical evidence indicates that real-world demands often exhibit this property in many areas of the globe. Intuitively, this phenomenon can be attributed to the clustering of major economic and population centers. Currently, there are an estimated 150 operational Starlink ground stations worldwide~\cite{starlinkgroundstationloc}, with the majority concentrated in North America, South America, Europe, and the Australia/Oceania region. Correspondingly, most user traffic originates from North America, Europe, and the Asia-Pacific region~\cite{li2025small}, based on Cloudflare's public measurements~\cite{cloudflareRadar}. Furthermore, even when considering rural users, the fraction of the Earth’s surface inhabited by humans remains relatively small; less than 15\% of the global land area is influenced by human activity, with urban and agricultural regions comprising less than half of that area~\cite{percentagelandarea}.

To model a futuristic demand scenario, we selected 100 highly populated cities on Earth and constructed a pairwise traffic pattern based on both population and distance, where longer distances correspond to potentially greater latency reduction~\cite{chaudhry2020free} and reflect regional isolation. Figure~\ref{motivation:demandGeometry:pattern} shows the geodesic flows along with their traffic intensities (left). Striking patterns emerge over the North Pacific, connecting East Asia to North America, as well as a noticeable pattern linking Australia to North America. To quantify these patterns, we partitioned the Earth into smaller regions using a latitude–longitude mesh. We followed the rule of thumb of $l_{\theta} = l_{\varphi} \cos(45^\circ)$, where $l_{\varphi}$ is the longitude step, and $l_{\theta}$ is the latitude step, to achieve a more uniform surface area distribution across regions~\cite{trent2023parameterization}. For each region, the aggregate direction is the sum of the unit tangent vectors (Equation~\ref{appendix:geomCalc:field:tangent}), at a point in the region, of all geodesic flows passing through the region (Appendix~\ref{appendix:geomCalc:geodesicRegion}), each weighted by its traffic intensity. This aggregate direction is represented by two components: the latitude basis $\hat{\theta}$ and longitude basis $\hat{\varphi}$ (Appendix~\ref{appendix:geomCalc:latlon}), as illustrated in Figure~\ref{motivation:demandGeometry:pattern} (right). To quantify the degree of directional clustering, which highlights the presence of global geometrical traffic patterns, we employ the \textbf{weighted mean resultant length}, a metric from circular statistics. This metric ranges from 0 to 1, where 0 indicates complete randomness and 1 indicates perfect alignment. For the demand considered, the global metric value was measured at 0.72 for $l_{\varphi} = 30^\circ$ and $l_{\theta} = 20^\circ$, indicating a strong and discernible directional pattern.

\section{Starfield Design}\label{design}

\subsection{Overview}\label{design:overview}
Starfield is an algorithm that achieves near-optimal topology performance by selecting a $\kappa$-max-degree subgraph $G$, which represents the network topology, from a set of satellites $S$ based on the corresponding in-range satellite graph $R$ and demand matrix $\Delta$. Intuitively, it constructs a highway-like topology—prioritizing major traffic corridors—rather than a dense, grid-like layout under limited link budgets. Starfield orients links toward major demand flows by following geodesic curves. Toward this objective, Starfield derives an abstract vector field from the geodesic demands. It then defines a Riemannian metric on the shell manifold to assign a distance to each link. The metric captures the influence of individual demand flows by elongating distances perpendicular to the flow's direction and contracting distances along the flow. Finally, to achieve near-uniform link coverage in all directions such as in +Grid, Starfield uses the Riemannian metric to choose a primary ISL at each satellite, from which other ISLs are derived. E.g., if $\kappa = 4$, each satellite chooses one ISL through the Riemannian metric, with the other ISL computed to be perpendicular to the first ISL. 

\subsection{Abstract Demand-based Vector Field}\label{design:field}
\begin{figure}[!t]
  \centering
  \begin{minipage}[t]{0.48\linewidth}
    \centering
    \includegraphics[width=\linewidth]{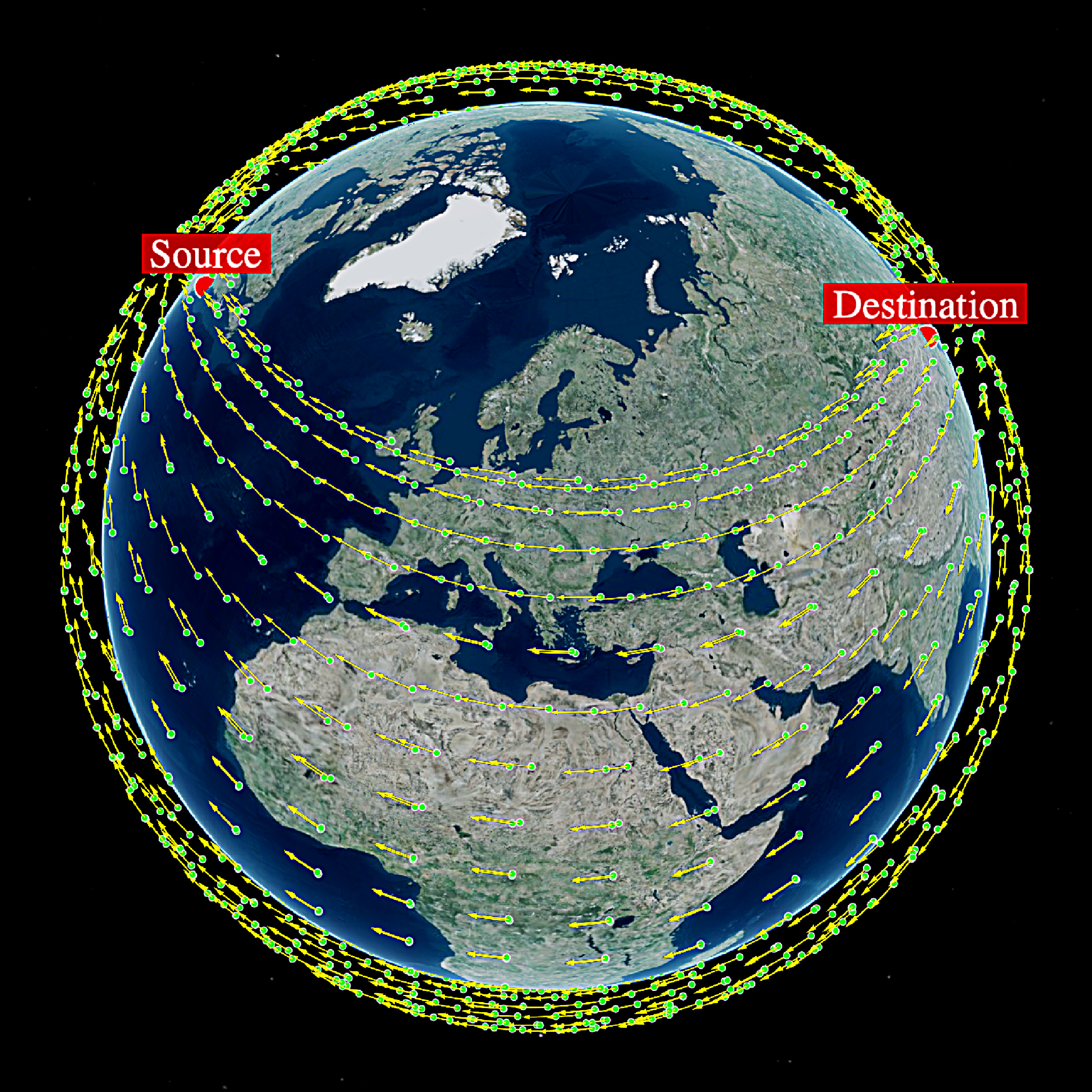}
  \end{minipage}\hfill
  \begin{minipage}[t]{0.48\linewidth}
    \centering
    \includegraphics[width=\linewidth]{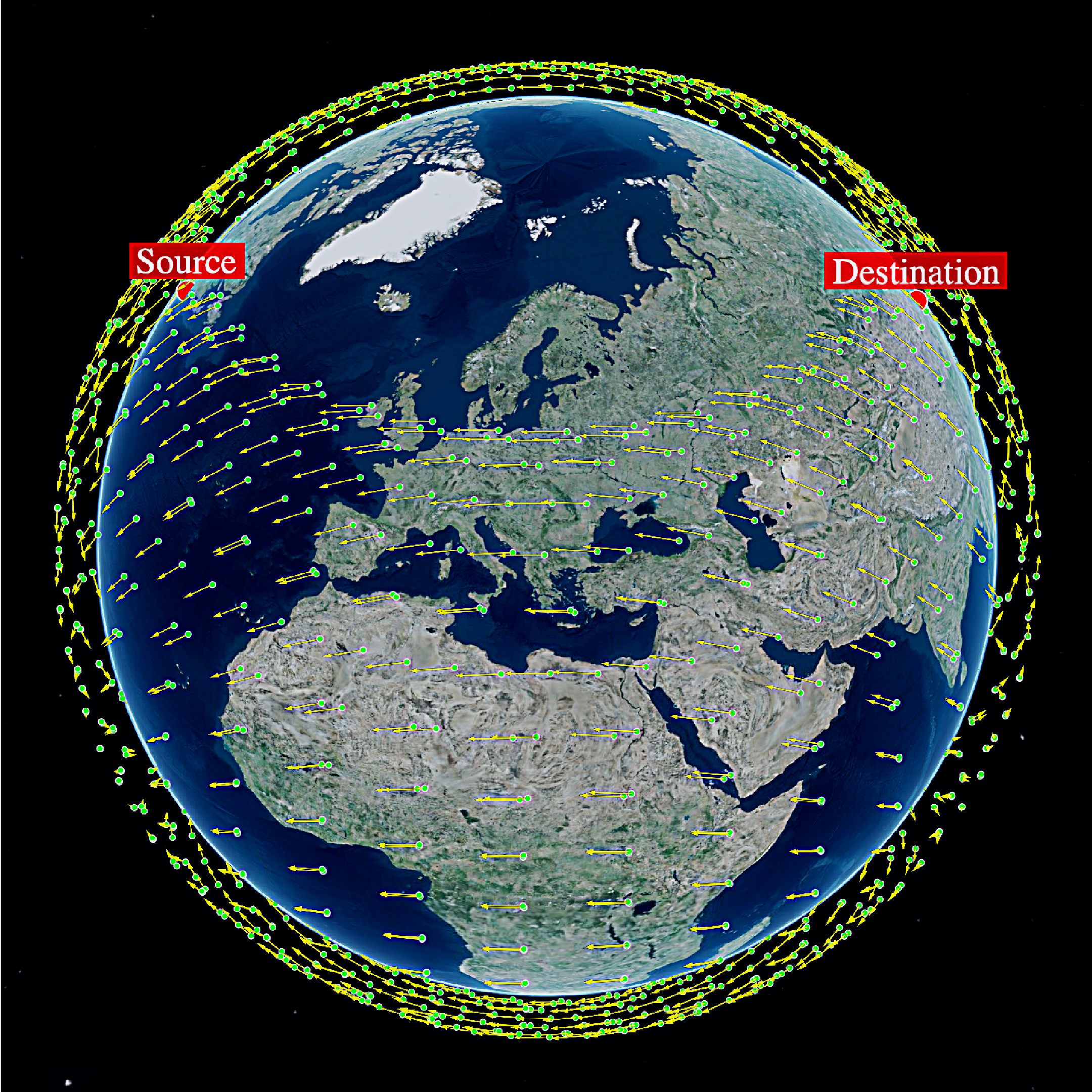}
  \end{minipage}
  \caption{Visualization of the electric vector field (left) and the proposed vector field (right) on the spherical shell under the influence of traffic flow between between a source and destination. Vector magnitudes are ignored for clarity at each satellite.}
  \label{design:field:FieldsVis}
\end{figure}
We propose an abstract vector field defined over the satellite shell to guide the orientation of inter-satellite links toward the geodesic connecting each source–destination pair, as the geodesic represents the theoretical minimum-distance path that packets can traverse. Based on our observations, selecting vector field functions oriented toward the geodesic curve—rather than directly toward the source or destination—yields shorter routing paths. In addition, the magnitude of the vector field should capture both the proximity of a point on the shell to the source or destination and the relative strength of the flow in terms of traffic volume. Incorporating these considerations, the vector field at a point $p$, with position vector $\vec{p}$, corresponding to flow $i$ with source $u_i$ and destination $v_i$, is defined as follows:
\begin{equation}\label{design:field:FieldEquation}
    \vec{f_i}(\vec{p}, \Delta_{i}) = K\Delta_{i} \left[ \frac{\hat{\tau}^p_{pu_i}}{d^2_{\mathbf{S}^2}(p,v_i)} - \frac{\hat{\tau}^p_{pv_i}}{d^2_{\mathbf{S}^2}(p,u_i)} \right]
\end{equation}
where $K$ is a constant, $\Delta_i$ denotes the traffic volume of flow $i$, $\hat{\tau}^p_{pu_i}$ and $\hat{\tau}^p_{pv_i}$ represent the unit tangent vectors of the geodesic curves from $p$ to $u_i$ and from $p$ to $v_i$, respectively, evaluated at $p$. Furthermore, $d^2_{\mathbf{S}^2}(p,u_i)$ and $d^2_{\mathbf{S}^2}(p,v_i)$ denote the squared geodesic distances between $p$ and $u_i$, and between $p$ and $v_i$, on the spherical manifold ($\mathbf{S}^2$). To simplify geometric calculations, we project ground stations from the Earth’s surface onto the satellite shell by scaling their position vectors. Specifically, for each ground station, the new position vector is obtained by multiplying its original vector by $\frac{\rho}{\rho_E}$, where $\rho$ is the shell radius, and $\rho_E$ is the Earth's radius (for details on geodesic and unit tangent vector calculations, refer to Appendix~\ref{appendix:geomCalc:field}). As shown in Equation~\ref{design:field:FieldEquation}, the field function couples the tangent vector ($\hat{\tau}^p_{pu_i}$) of one geodesic ($p-u_i$) with the squared geodesic distance $d^2_{\mathbf{S}^2}(p,v_i)$ of the other geodesic ($p-v_i$). This interaction emphasizes the direction toward the destination for points surrounding the source, and vice versa, thereby orienting the vector field along the geodesic connecting source to destination. Equation~\ref{design:field:FieldEquation} is inspired by the electromagnetic vector field and becomes directly analogous to it when $d^2_{\mathbf{S}^2}(p,u_i)$ and $d^2_{\mathbf{S}^2}(p,v_i)$ are interchanged. Figure~\ref{design:field:FieldsVis} compares the directional behavior of the electromagnetic field (left) and the proposed field (right) at satellite locations for a traffic flow between a source and destination.

\subsection{Starfield Algorithm}\label{design:algorithm}

Algorithm~\hyperref[design:overview:algorithm]{1} presents the Starfield topology generation algorithm. Starfield assigns priorities to candidate links based on a distance metric derived from an abstract vector field (Section~\ref{design:field}) that captures the influence of individual traffic flows, as shown in Lines~3–10. The metric consists of a dominant component, $f^{\perp}{f^{\perp}}^T$, which orients links along the vector field, and a small regularization term, $\epsilon I$, which prevents degeneracy and preserves sensitivity to the original geometric distance in regions where the field magnitude is extremely weak (Line~7). Distances, for priority, can be computed as described in Section~\ref{background:riemannian}. However, due to the prohibitive computational cost of evaluating the associated integral, we approximate the distance by $|\vec{f^{\perp}} \cdot (\vec{P_s} - \vec{P_{s'}})|$, where $\vec{P_s}$ and $\vec{P_{s'}}$ denote the position vectors of satellites $s$ and $s'$ forming the link (Line~8).

Inspired by +Grid and aiming to construct a symmetric topology, Starfield builds the graph $G$ based on the distances induced by the proposed metric, as shown in Lines~11–19. For each satellite, Starfield first selects the closest neighbor by minimizing the aggregate distance over all traffic flows (Line~13). Given this minimum-distance link, Starfield then selects the remaining $K = \left\lfloor\frac{\kappa}{2}\right\rfloor - 1$ links by uniformly distributing them in angle, using a step size of $\frac{\pi}{K+1}$ (Line~16). The precise angular link selection procedure $\zeta$ is detailed in Appendix~\ref{appendix:geomCalc:angular}. The reason only half of the maximum allowed links are explicitly selected by each satellite is that each satellite is assumed to receive the other half of its connections from neighboring satellites. Note that during the link selection procedure, links may be redundantly chosen, reducing overall link utilization; this effect can be mitigated by considering the second-closest or subsequent neighbors when establishing additional links.

\begin{center}
\begin{minipage}{\columnwidth}
\captionsetup{type=algorithm}

\begin{mdframed}[
  topline=true,
  bottomline=true,
  leftline=false,
  rightline=false,
  linewidth=0.6pt,
  skipabove=6pt,
  skipbelow=6pt,
  innertopmargin=6pt,
  innerbottommargin=6pt
]
\textbf{Algorithm 1.} Starfield topology generation algorithm \\[-0.5em]
\label{design:overview:algorithm}
\hrule height 0.6pt
\vspace{0.8em}

\begin{algorithmic}[1]
\Require Satellite set $S$, ISLs per satellite $\kappa$, Demand matrix $\Delta_{m \times m}^T$, Satellites' positions $P^T$, In range satellites graph $R^T$, Time period $T$, Shell radius $\rho$
\Ensure Topology $G$

\State $G \, ,D^{m \times m}_{n \times n} \leftarrow \emptyset \,, 0$
\State $\Delta \,, R \,, P \leftarrow \frac{1}{T}\sum_{t \in T} \Delta^t \,, \bigcap_{t \in T} R^t \,, P^{\frac{T}{2}}$
\State /****************** \textbf{Per-flow metric-based link distance calculation} *****************/
\ForAll{$(s,s') \in R$}
    \ForAll{$(u,v) \in \Delta$}
      \State $\vec{f_{uv}} \leftarrow \vec{f}_{uv}(\vec{P_s}\,,\,\Delta_{uv})$ \qquad // Section~\ref{design:field}
      \State $\vec{f^{\perp}_{uv}} \leftarrow \frac{\vec{f_{uv}} \times \vec{P_s}}{\rho} $ 
      \State $g^s_{uv} \leftarrow f_{uv}^{\perp}{f_{uv}^{\perp}}^T + \epsilon I$
      \State $D^{uv}_{ss'} \leftarrow \int_0^{1} \sqrt{g^{\lambda}_{uv}\dot{\gamma}^{\lambda}_{ss'}{g^{\lambda}_{uv}}^T}\,d\lambda \approx |\vec{f^{\perp}_{uv}} \cdot (\vec{P_s} - \vec{P_{s'}})| $
    \EndFor
\EndFor
\State /*********************** \textbf{Distance-based subgraph selection} **************************/
\State $K \leftarrow \left\lfloor\frac{\kappa}{2}\right\rfloor - 1$
\ForAll{$s \in S$}
    \State $s^* \leftarrow \min_{s'} \sum_{u,v \in \Delta}D^{uv}_{ss'}$
    \State $G \leftarrow G \cup \{ss^*\}$
    \ForAll{$j = 1 \to K$}
        \State $s'' \leftarrow \zeta(\vec{P_s} \,, \vec{P_{s*}} \,, j\frac{\pi}{K+1})$ \qquad // Appendix~\ref{appendix:geomCalc:angular}
        \State $G \leftarrow G \cup \{ss''\}$
    \EndFor
\EndFor
\State \Return $G$
\end{algorithmic}
\end{mdframed}
\end{minipage}
\end{center}

\subsection{Additional Considerations}\label{design:considerations}
\subsubsection{Link Length Prioritization}\label{design:considerations:priority}
In Section~\ref{design:overview}, we approximated the metric as $\vec{f^{\perp}_{uv}} \cdot (\vec{P_s} - \vec{P_{s'}})$, where $\vec{P_s}$ and $\vec{P_{s'}}$ denote the position vectors of satellites $s$, and $s'$, respectively, and $\vec{f^{\perp}_{uv}}$ is the perpendicular vector to the field induced by the source–destination pair $(u,v)$ at $s$ (Section~\ref{design:field}). While the normalized term $\vec{f^{\perp}_{uv}} \cdot \frac{\vec{P_s} - \vec{P_{s'}}}{\|\vec{P_s} - \vec{P_{s'}}\|}$ controls link orientation, the metric formulation inherently favors minimizing $\|\vec{P_s} - \vec{P_{s'}}\|$, thereby prioritizing shorter links. Although shorter links facilitate closer adherence to geodesic paths and thus reduce stretch factor (Section~\ref{evaluation:setup:metrics}), this bias comes at the expense of increased hop count (Section~\ref{evaluation:setup:metrics}), as longer links that could reduce the number of hops are disfavored.

To introduce a trade-off between stretch factor and hop count, we bias the metric minimization toward longer links by modifying the distance term to $\vec{f^{\perp}_{uv}} \cdot \frac{\vec{P_s} - \vec{P_{s'}}}{\|\vec{P_s} - \vec{P_{s'}}\|^2}$. However, uniformly prioritizing longer links—particularly for satellites near the source or destination—can lead to hops that overshoot these endpoints, potentially pushing the path outside the ground–satellite link range. This behavior unintentionally degrades both hop count and stretch factor. To address this issue, we incorporate $\|\vec{f_{uv}}\|$ as a proxy for proximity to the source or destination, and use it to determine whether shorter or longer links should be favored. Specifically, satellites experiencing stronger field intensity (i.e., closer to the endpoints) prioritize shorter links using the distance term  $\vec{f^{\perp}_{uv}} \cdot (\vec{P_s} - \vec{P_{s'}})$ , whereas satellites in weaker field regions favor longer links using the modified distance $\vec{f^{\perp}_{uv}} \cdot \frac{\vec{P_s} - \vec{P_{s'}}}{\|\vec{P_s} - \vec{P_{s'}}\|^2}$. We employ the exponential weighting function $e^{-\|\vec{f_{uv}}\|}$ to smoothly interpolate between these two regimes. This function approaches 1 in regions where the field is weak and decays toward 0 as the field strength increases. Combining these components, the modified distance metric is defined as:
\begin{equation}\label{design:considerations:priority:formula}
    D^{uv}_{ss'} = \vec{f^{\perp}_{uv}} \cdot \frac{\vec{P_s} - \vec{P_{s'}}}{\|\vec{P_s} - \vec{P_{s'}}\|^{2e^{-\|\vec{f_{uv}}\|}}}
\end{equation}
The trade-off can be manually controlled by scaling the overall field achieved by adjusting the constant $K$ in Equation~\ref{design:field:FieldEquation}. Increasing $K$ amplifies the influence of the field, thereby favoring paths that more closely follow geodesic curves and reducing stretch factor, whereas decreasing $K$ weakens the field influence and encourages longer links, favoring lower hop count. Section~\ref{evaluation:results:ablation} evaluates the impact of $K$ on Starfield’s performance through ablation experiments.

\subsubsection{Dealing with the Crown}\label{design:consideration:crown}
As shown in Figure~\ref{motivation:demandGeometry:pattern} (left), a large fraction of traffic geodesics between highly populated cities pass through the North Pole region. However, due to the limited orbital inclination, Phase 1 Starlink satellites do not cover this area. This demand pattern highlights the importance of link orientation near the boundary of satellite coverage, which we refer to as the \textit{crown}. The absence of satellite coverage in the polar region, which is inherently tied to the constellation’s inclination, leads to suboptimal routing paths in Starfield. As noted in~\cite{alexander1981geodesics} and~\cite{de2008computational}, the optimal path on a Riemannian manifold with boundary—the constellation's shell—follows a geodesic within the interior of the manifold, but aligns with the boundary curve at points where the geodesic intersects the boundary. Therefore, to achieve near-optimal paths under this constraint, links in the crown area should be oriented parallel to the crown's line of latitude, rather than forming a jagged pattern, which may arise from the original vector field design (Section~\ref{design:field}). To this end, we modify the vector field ($\vec{f_c}$) to gradually orient links toward lines of latitude as satellites approach the crown:
\begin{equation}\label{design:considerations:crown:formula}
    \vec{f_c} = \vec{f}\, + \, \eta e^{-\omega(\sin{i} - \frac{|\vec{p} \cdot \hat{z}|}{\rho})}(\vec{f} \cdot \hat{\varphi_p}) \hat{\varphi_p}
\end{equation}
where $\vec{p}$ is the position vector of the point $p$ on the shell we are manipulating field for, $i$ is the inclination, $\rho$ is the shell radius, $\hat{z}$ is unit vector at direction $Z$, $\hat{\varphi_p}$ (Appendix~\ref{appendix:geomCalc:latlon}) is the latitude unit vector at point $p$. $\eta$ is a constant controlling reorientation strength toward latitude lines, $\omega$ is a constant controlling the vicinity of the crown under the influence of reorientation. In Figure~\ref{design:field:FieldsVis}, we observe that links in the crown region—particularly north of Europe—are oriented along lines of latitude as a result of applying Equation~\ref{design:considerations:crown:formula}.

\subsubsection{Time Complexity}
The complexity consists of two parts: distance computation and subgraph selection. Since vector operations are constant-time, the per-flow link distance computation has complexity $\mathcal{O}(m^2|E(R)|)$, where $|E(R)|$ denotes the number of edges in $R$, and $m$ is the number of ground stations. Due to $\zeta$ composition of constant-time vector operations, and on-the-fly minimum distance computation, subgraph selection has complexity $\mathcal{O}(\kappa|S|)$. Therefore, the overall complexity can be combined to $\mathcal{O}(m^2|E(R)| + \kappa|S|)$.

\subsubsection{Analysis}
To simplify the mathematical model, a setting is assumed in which the manifold is a two-dimensional plane equipped with the standard Euclidean metric, and satellites remain stationary over a small time window. In this setting, we derive a lower bound on the stretch for a given traffic demand pattern under all possible topologies. We also show an upper bound on the stretch achieved by the Starfield topology.
\begin{theorem}[informal]
For any demand, the stretch achieved in Starfield is upper bounded by a factor that depends on the sum of the sine and cosine of the angles between the demand's geodesic and the vector field along the geodeisc. 
\end{theorem}
For demands that do not have other competing demands nearby, we show that Starfield's stretch is near-optimal. Refer to Appendix~\ref{appendix:analysis} for the precise mathematical theorems and proofs.

\section{Evaluation}\label{evaluation}
\subsection{Experimental Setup}\label{evaluation:setup}

\subsubsection{Simulator}\label{evaluation:setup:simulator}
To model a LEO satellite network, we implemented a custom packet-level, network-layer simulator, designed to be easier to use, faster, lighter, and more configurable—particularly in terms of inter-satellite topology—than existing state-of-the-art simulators~\cite{kassing2020exploring},~\cite{lai2020starperf} and emulators~\cite{lai2023starrynet}. To achieve link awareness, the simulator explicitly accounts for time-varying distances between satellites and ground stations, which affect achievable data rates through fluctuations in the signal-to-noise ratio. Specifically, the simulator computes link capacity using the Shannon–Hartley theorem in conjunction with the inverse-square law governing electromagnetic wave propagation. This model directly influences the input/output queue processing in simulator's network interfaces. 
\begin{table}[!t]
\centering
\begin{minipage}[t]{0.48\linewidth} 
\centering
\textbf{(a) Constellation Parameters} \\[0.3em] 
\begin{tabular}{lc}
\toprule
Parameter & Value \\
\midrule
Altitude & \SI{550}{\kilo\metre} \\
Inclination & \SI{53}{\degree} \\
Number of orbits & 72 \\
Number of satellites per orbit & 22 \\
Minimum elevation & \SI{25}{\degree} \\
Mean revolution speed & \SI{3.98}{\radian\per\hour} \\
\bottomrule
\end{tabular}
\end{minipage}\hfill
\begin{minipage}[t]{0.48\linewidth} 
\centering
\textbf{(b) Network Parameters} \\[0.3em] 
\begin{tabular}{lc}
\toprule
Parameter & Value \\
\midrule
ISL Bandwidth & \SI{1}{\tera\bit\per\second} \\
GSL Bandwidth & \SI{100}{\giga\bit\per\second} \\
ISL Noise Coefficient & 0.1 \\
GSL Noise Coefficient & 0.001 \\
Interface Buffer Size & 1000 pkts \\
Packet Size & \SI{12}{\kilo\byte} \\
\bottomrule
\end{tabular}
\end{minipage}
\caption{Simulation parameters: (a) Constellation and (b) Network}
\label{evaluation:setup:simulator:params}
\end{table}
Table~\ref{evaluation:setup:simulator:params} summarizes the simulation parameters. We modeled the Phase 1 Starlink shell with 1,584 satellites, with adjacent orbits phased by half an angular step. To avoid atmospheric interference, the maximum ISL length was limited so that no link extends below an altitude of $80\,\text{km}$. Each satellite was limited to a maximum of four ISLs at any point in the simulation, constraining the Starfield and Random topologies to enable a fair comparison with +Grid. Furthermore, Packets are routed using Dijkstra’s shortest-path algorithm~\cite{dijkstra2022note}, restricted to ISLs and forwarded according to precomputed ground station destined shortest-path tables at each node. The simulator was executed on a server equipped with an Intel® Xeon® W-2295 CPU @ $3.00\,\text{GHz}$ and $120\,\text{GB}$ of available RAM. Routing half a million packets required approximately an hour of wall-clock time. To generate results efficiently, we allocated 10 seconds for brief, exploratory experiments and 100 seconds for longer runs, capturing realistic network dynamics. The source code consists of more than 9,000 lines implemented in Go, C/C++, Python, JavaScript, and HTML/CSS, and to support reproducibility, the simulator source code is made available via an anonymous GitHub repository.\footnote{\url{https://github.com/shayunak/Starfield}}

\subsubsection{Baselines}\label{evaluation:setup:baselines}
To describe the baseline topologies, we note that satellites are uniformly distributed within each orbit, and the orbits themselves are evenly spaced over the shell. Accordingly, each satellite can be denoted by $S_i^j$, where $i$, and $j$ are sequential identifiers for the orbit and the satellite within that orbit, respectively. Let $N_O$ and $N_S$ denote the number of orbits and the number of satellites per orbit. To evaluate the performance of Starfield, we construct the following baseline topologies for comparison:
\begin{itemize}
  \item \textbf{+Grid}: Each satellite $S_i^j$ establishes links with $S_i^{j+1 \bmod N_S}$, $S_i^{j-1 \bmod N_S}$, $S_{i+1 \bmod N_O}^j$, and $S_{i-1 \bmod N_O}^j$.
  \item \textbf{Random}: To ensure inter-orbital consistency in a static random topology, we select a random feasible inter-orbital linking pattern between adjacent orbits, following the approach described in Section~\ref{evaluation:results:static}. Accordingly, we construct the static random topology by randomly partitioning the number of allowable ISLs into intra-orbital and inter-orbital links. For intra-orbital connectivity, each satellite within an orbit establishes a link with another in-range satellite from the same orbit. For inter-orbital connectivity, half of the links are assigned to the next adjacent orbit, with the remaining half implicitly provided by the previous adjacent orbit. For each adjacent orbit pair, multiple random feasible patterns—subject to the constraint that all satellites following a pattern remain within communication range—are considered. Thus, for each inter-orbital pattern $p$, each satellite $S_i^j$ in the orbit $i$ establishes a link with $S_{i+1 \bmod N_O}^{j+p \bmod N_S}$.
  \item \textbf{Motif}: The Motif-based topology scheme~\cite{bhattacherjee2019network} formulates topology selection as a search problem driven by a linear objective function. Their results indicate that most of the discovered motifs are heavily biased toward reducing hop count, largely neglecting stretch factor, and consequently \textbf{fail} to outperform +Grid. This behavior suggests an ineffective formulation of the trade-off objective function, stemming from the direct aggregation of metrics with fundamentally different natures—one being an absolute quantity and the other a normalized factor. This limitation is significant, as stretch factor is directly correlated with propagation delay, a dominant component of end-to-end network latency. Additionally, the motif-based approach constrains the topology search space to uniform motif structures, thereby limiting the range of achievable optimal solutions. Nevertheless, we use $\Phi_1$, the sum of hop count and stretch factor, as the objective function to optimize the motif search.
  \item \textbf{LAMP}: LAMP~\cite{esswein2025lamp}, a state-of-the-art dynamic satellite topology, employs a hybrid scheme of persistent and temporary inter-satellite links to achieve lower mean communication latency than static grid networks. The persistent topology is based on +Grid topology, where one of the four per-satellite links is removed, so the free link would be dedicated to dynamic topology based on the all pairs shortest path minimization objective.
\end{itemize}

\subsubsection{Traffic Demands}\label{evaluation:setup:demands}
For end-to-end traffic generation, we selected 100 highly populated cities~\cite{CityPopulation} and their corresponding locations~\cite{CityLocation}, and created pairwise flows with uniform, hotspot-oriented, distance-based, and population-based patterns. While no global data-demand model exists, population provides a reasonable proxy for city-to-city traffic. ISL routing improves latency between distant locations~\cite{chaudhry2020free}, and isolated cities—those farther from other population centers on average—require higher traffic capacities. Consequently, distance-based traffic patterns offer another realistic city-to-city data-demand. We model traffic demand at each time interval as a matrix ($\Delta_{ij}$). Each demand matrix consists of two components: (1) a constant component, where each entry is drawn from a uniform distribution bounded by the buffer size to prevent overloading of ground–satellite links, and (2) a normalized weight component $W$. The weight $W_{ij}$ varies across demand patterns, where $i$ and $j$ are source and destination ground stations. Specifically, $W_{ij}$ is defined as 1, $e^{-\frac{i+j}{m}}$, $d_{ij}$, and $p_ip_j$ for uniform, hotspot-oriented, distance-based, and population-based traffic patterns, respectively. Here, $m$ denotes the number of ground stations, $d_{ij}$ the geodesic distance between $i$ and $j$, and $p_i$ the population associated with $i$.

\subsubsection{Performance Metrics}\label{evaluation:setup:metrics}
The performance of Starfield is assessed using the following network metrics:
\begin{itemize}
    \item \textbf{Stretch factor}: The total distance traversed by a packet—including ISLs and the two ground–satellite links—divided by the geodesic distance between the source and destination. This metric reflects normalized propagation delay.
    \item \textbf{Hop count}: The total number of links a packet traverses between source and destination.
    \item \textbf{Round trip time}: The duration between a packet’s initial transmission from a source and its return to the source from the destination.
    \item \textbf{Link Usage Ratio}: The fraction of routed packets in the network that are forwarded over an ISL.
    \item \textbf{Jitter}: The variability of the round-trip time, measured as the standard deviation of RTT across packets from a source to destination.
\end{itemize}

\subsection{Results}\label{results}

\subsubsection{Static Starfield}\label{evaluation:results:static}
Maintaining stability over time needs addressing inter-orbital distance variations. In \textit{static Starfield}, we consider a stable and homogeneous topology in which ISLs between satellite pairs are always in range and, therefore, do not require any reconfiguration over an orbital period. We enforce homogeneity by following the same local ISL connectivity pattern on all satellites within an orbit. We compute the static Starfield topology by selecting for each adjacent orbit pair a feasible connection pattern that minimizes the total induced link distances. Formally, assigning indices $O_i$ to satellites in orbit $O$, the pattern $p$ is defined as:
\begin{equation}\label{evaluation:results:static:pattern}
    p^{*} = \min_{p} \sum_f\sum_{j=1}^{N_S} D_f(O_j, O'_{j+p \bmod N_S})
\end{equation}
where $N_S$ denotes the number of satellites per orbit, $D_f$ the flow distance function (Section~\ref{design:overview}), and $O, O'$ the target orbits.

\begin{figure}[!t]
    \centering
    \includegraphics[width=\columnwidth]{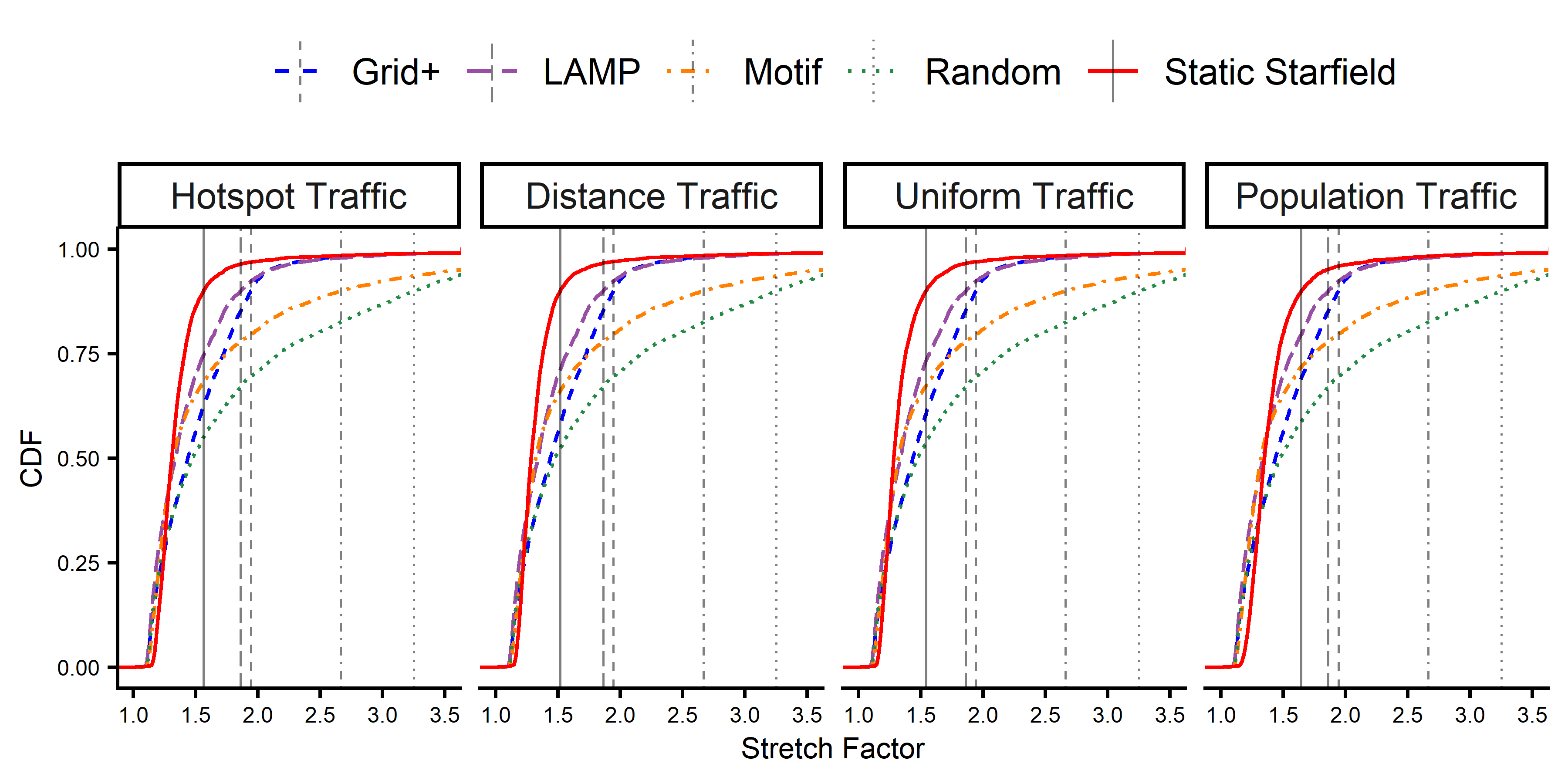}
    \caption{CDFs of city-to-city stretch factor for static Starfield, Motif, +Grid, and Random. Vertical lines indicate the $90^{\text{th}}$ percentile.}
    \label{evaluation:results:static:StretchCDFStatic}
\end{figure}

Figure~\ref{evaluation:results:static:StretchCDFStatic} shows the cumulative distribution functions (CDFs) of the stretch factor under different traffic demand patterns for the static Starfield and four baselines, +Grid, Motif, LAMP, and Random. Starfield consistently outperforms both baselines across all four demand patterns, achieving up to a 15\% reduction in stretch factor relative to LAMP, up to a 20\% relative to +Grid, up to a 40\% reduction relative to Motif, and up to a 50\% reduction relative to Random under distance-based demand. The performance gain diminishes as the traffic pattern becomes more uniform. In our simulation, highly populated cities are approximately uniformly distributed over the Earth, making the population-based demand pattern the most directionally uniform. In contrast, distance-based and hotspot traffic induce more concentrated, directional flows. As demand becomes increasingly uniform, Starfield’s advantage decreases, which can be attributed to flows arriving from a broader range of directions, making it infeasible to locally accommodate all directions efficiently.

\begin{figure}[!t]
    \centering
    \includegraphics[width=\columnwidth]{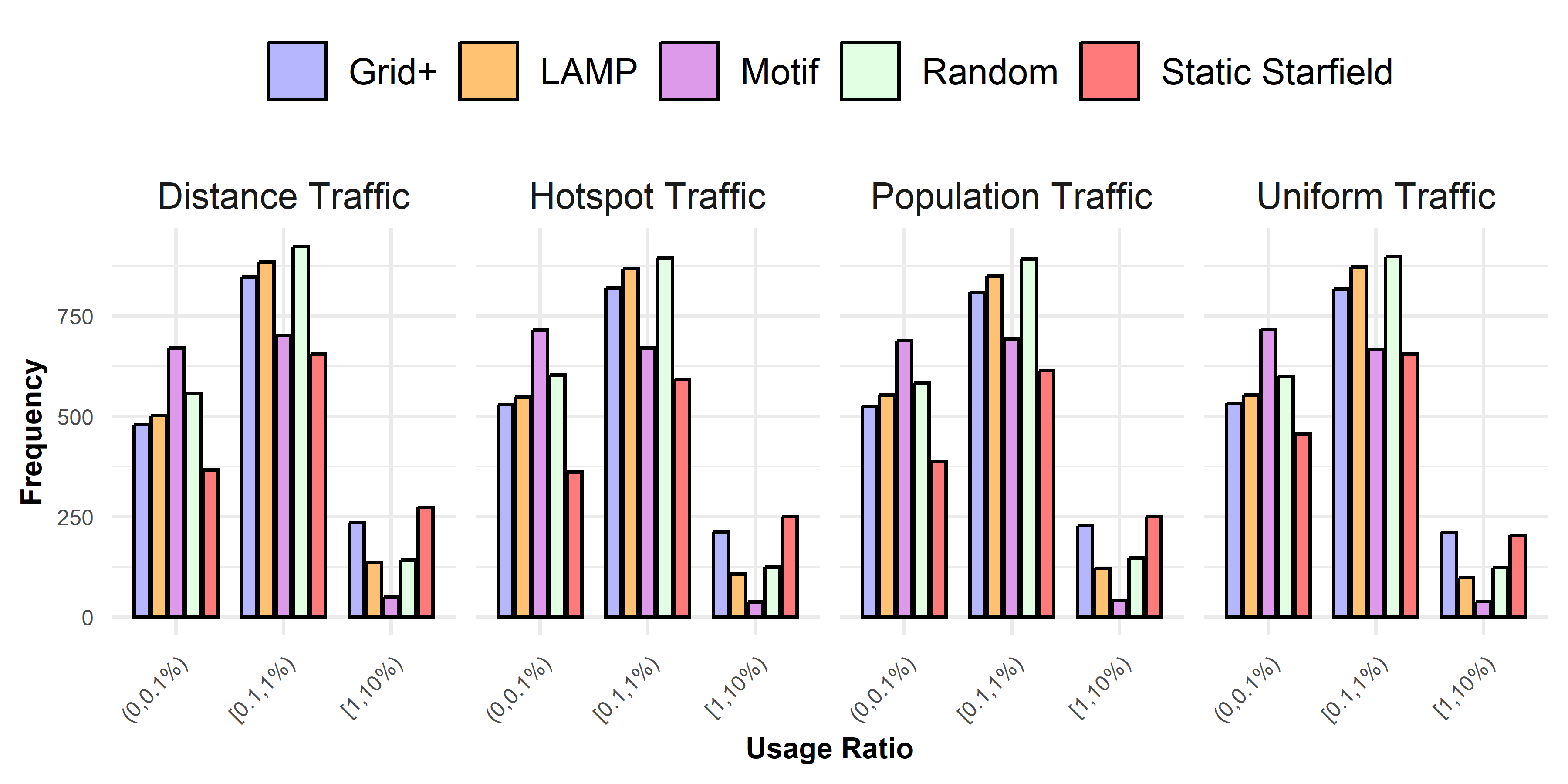}
    \caption{Link usage ratio histogram of static Starfield, Motif, LAMP, +Grid, and Random.}
    \label{evaluation:results:static:LinkUsageStatic}
\end{figure}

Figure~\ref{evaluation:results:static:LinkUsageStatic} presents link usage, measured by the proportion of total packets relayed on each link as an indicator of congestion, under different traffic demand patterns for the static Starfield topology, +Grid, LAMP, Motif, and Random baselines. Starfield exhibits higher link congestion across all traffic patterns, characterized by a larger fraction of heavily utilized links. This behavior arises because Starfield intentionally routes traffic along shortest paths to minimize stretch factor, which in turn concentrates traffic on these paths. In contrast, the Motif topology demonstrates lower congestion, likely due to the presence of multiple randomized shortest paths that distribute traffic more evenly across the network, and lower hop counts. Furthermore, more uniform traffic patterns result in reduced congestion overall, primarily due to the naturally balanced distribution of flows.

\subsubsection{Dynamic Starfield}\label{evaluation:results:dynamic}
To combine the complementary characteristics of population-based and distance-based demands and construct a more realistic traffic model, we merge these two patterns and conduct a longer experiment comparing dynamic Starfield, static Starfield, and +Grid over a 100-second duration. Dynamic Starfield applies the Starfield topology periodically. Specifically, we divide the 100-second interval into two 50-second epochs, a timescale over which link dynamicity becomes significant. We assume zero overhead for link reconfiguration, thereby allowing minor flexibility in topology adaptation.

\begin{figure}[!t]
  \centering
  \begin{minipage}[t]{0.48\linewidth}
    \centering
    \includegraphics[width=\linewidth]{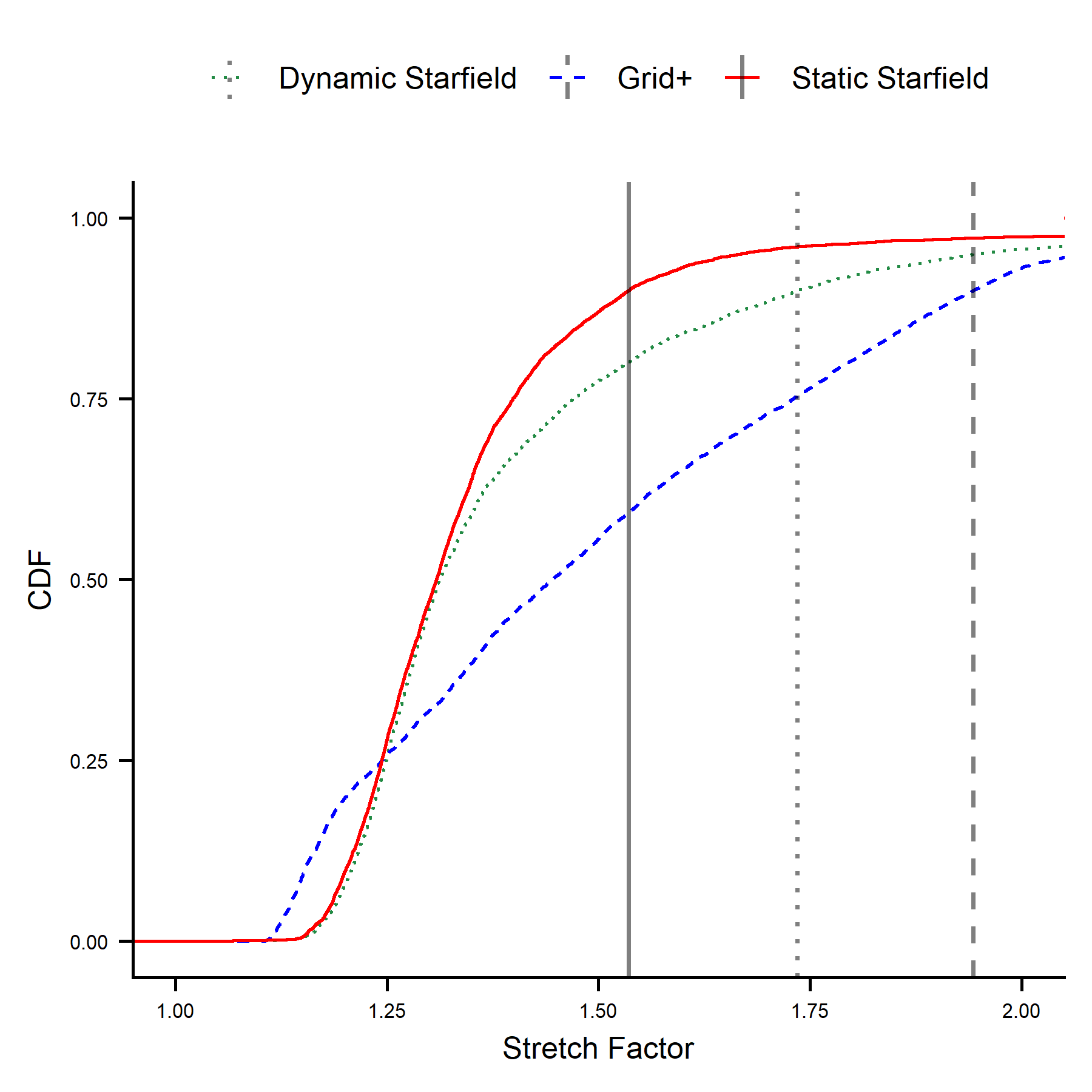}
  \end{minipage}\hfill
  \begin{minipage}[t]{0.48\linewidth}
    \centering
    \includegraphics[width=\linewidth]{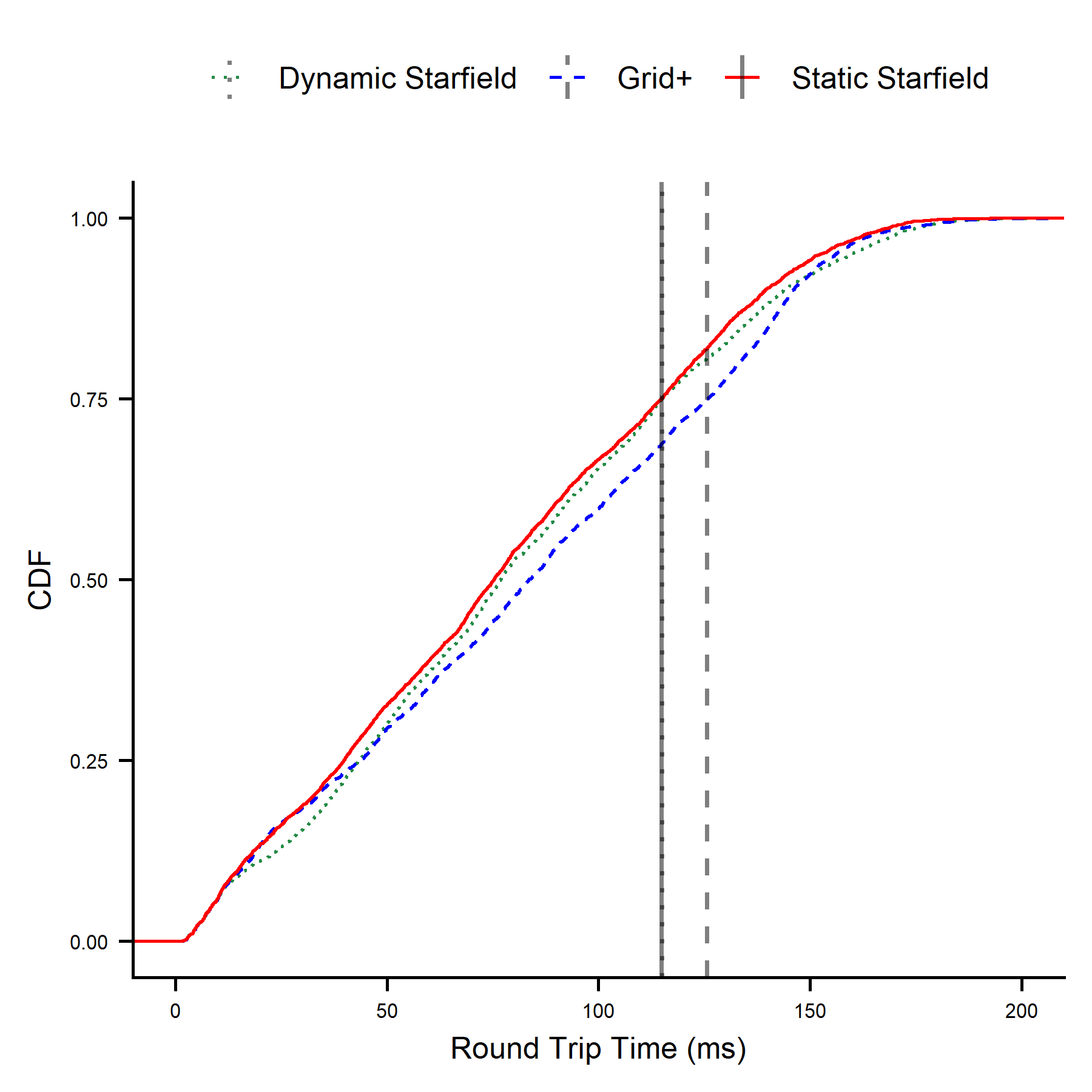}
  \end{minipage}
  \caption{CDFs of city-to-city stretch factor (left) and round trip time (right) for +Grid, dynamic Starfield, and  static Starfield. (Left) vertical lines indicate $90^{\text{th}}$ percentile, and (right) vertical lines indicate $75^{\text{th}}$ percentile.}
  \label{evaluation:results:dynamic:StretchRTTCDFDynamic}
\end{figure}
Figure~\ref{evaluation:results:dynamic:StretchRTTCDFDynamic} illustrates the CDFs of stretch factor (left) and round-trip time (RTT) (right). Both static and dynamic Starfield outperform +Grid in terms of stretch factor, with dynamic Starfield achieving up to a 15\% reduction and static Starfield achieving up to 20\% reduction relative to +Grid. However, dynamic Starfield underperforms its static counterpart, which can be attributed to temporal inconsistencies in inter-satellite distances. The figure also shows a modest improvement of approximately $10\,\text{ms}$ in RTT for both static and dynamic Starfield compared to +Grid. Although a larger RTT reduction might be expected from improved propagation delay, this gain is offset by increased transmission and queuing delays caused by higher congestion and longer hop counts, ultimately limiting the end-to-end RTT improvement.

\begin{figure}[!t]
  \centering
  \begin{minipage}[t]{0.48\linewidth}
    \centering
    \includegraphics[width=\linewidth]{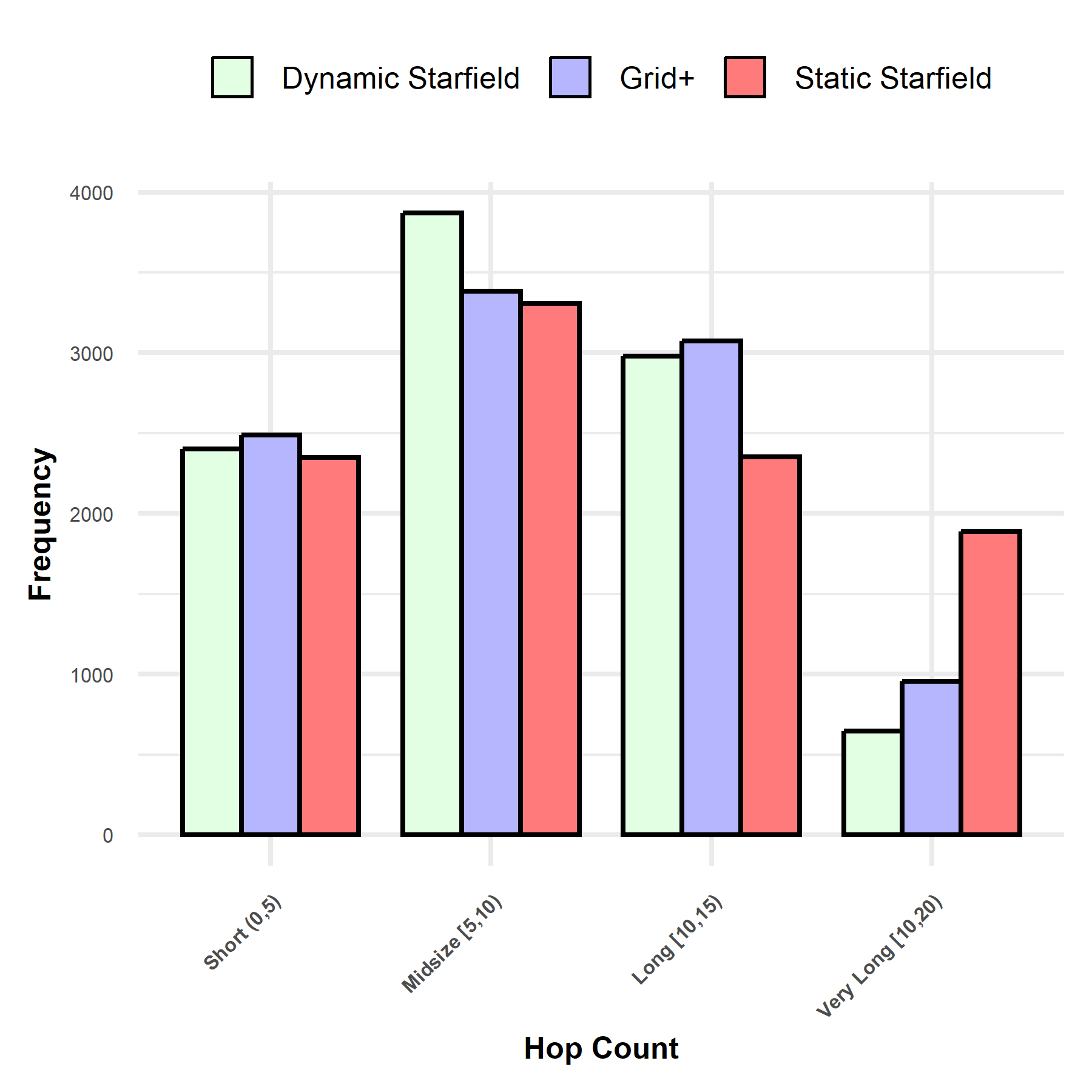}
  \end{minipage}\hfill
  \begin{minipage}[t]{0.48\linewidth}
    \centering
    \includegraphics[width=\linewidth]{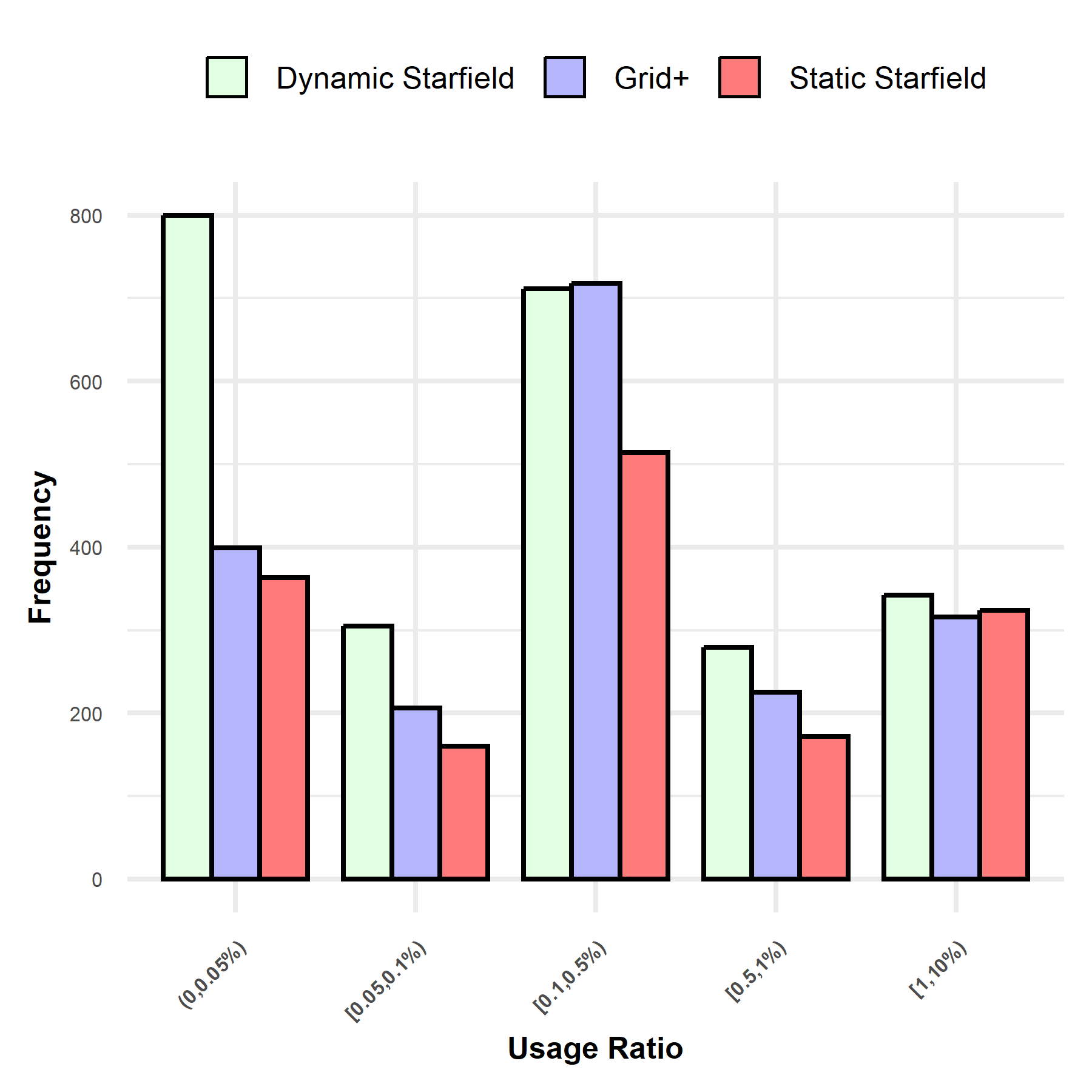}
  \end{minipage}
  \caption{Histogram of city-to-city hop count (left) and link usage ratio (right) for +Grid, dynamic Starfield, and  static Starfield. Hop counts (left) are divided into short, midsize, long, and very long paths.}
  \label{evaluation:results:dynamic:HopCountLinkUsageDynamic}
\end{figure}
Figure~\ref{evaluation:results:dynamic:HopCountLinkUsageDynamic} shows that dynamic Starfield achieves up to a 30\% reduction in hop count (left) relative to +Grid for very long paths, effectively redirecting traffic to medium-length paths of 5–10 hops. This improvement arises because the dynamic design can exploit longer inter-satellite hops, unconstrained by the inter-orbital pattern-matching limitations of the static topologies. In terms of link usage (right), dynamic Starfield exhibits higher utilization on previously lightly loaded links, as topological adaptation introduces new, less congested paths. However, heavily congested links largely persist, with dynamic Starfield showing a slightly higher frequency of highly used links, indicating that congestion is redistributed rather than eliminated.

\subsubsection{Crown Size Impact}\label{evaluation:results:crown}
\begin{figure}[!t]
  \centering
  \begin{minipage}[t]{0.48\linewidth}
    \centering
    \includegraphics[width=\linewidth]{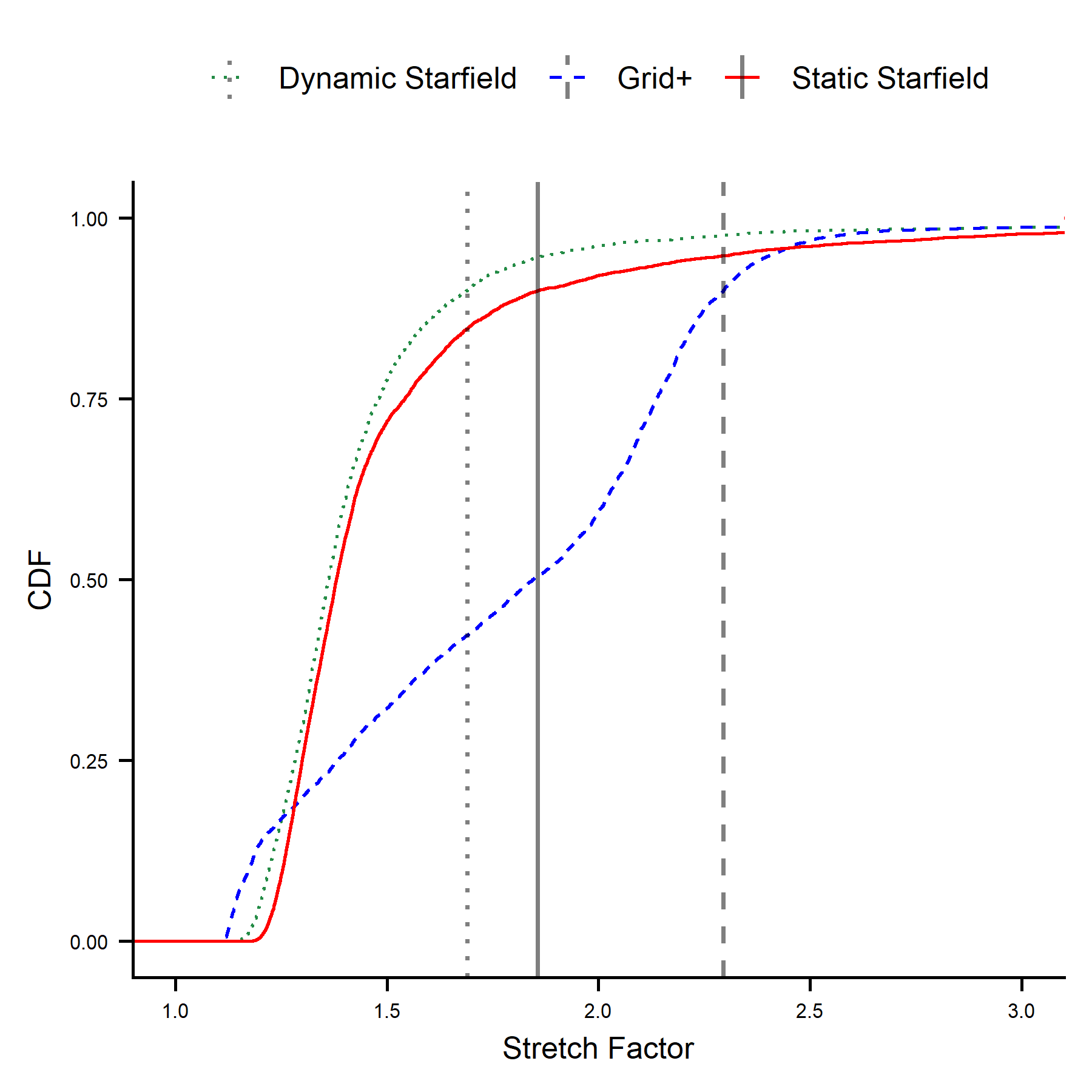}
  \end{minipage}\hfill
  \begin{minipage}[t]{0.48\linewidth}
    \centering
    \includegraphics[width=\linewidth]{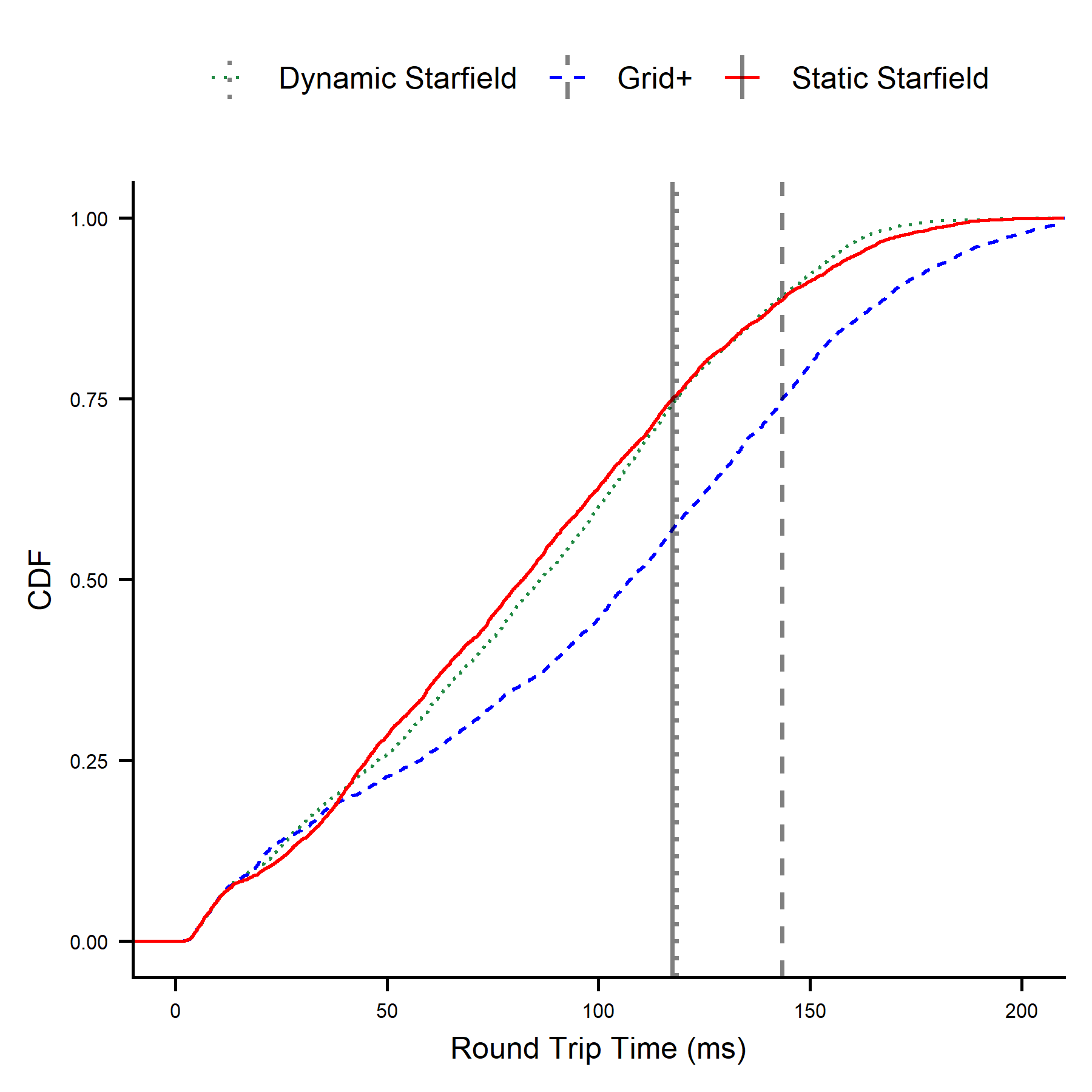}
  \end{minipage}
  \caption{CDFs of city-to-city stretch (left) and RTT (right) for +Grid, dynamic Starfield, and static Starfield under constellation inclinations of $80^\circ$; left lines mark the 90th percentile, right lines the 75th percentile.}
  \label{evaluation:results:dynamic:StretchRTTCDFDynamicHigherInclination}
\end{figure}
To evaluate the impact of crown size on Starfield’s performance, we increase the inclination from the Phase 1 Starlink value of $53^\circ$ to an arbitrarily higher $80^\circ$, thereby reducing the crown area and and enabling a larger fraction of geodesically optimal paths that are less affected by boundary.
Figure~\ref{evaluation:results:dynamic:StretchRTTCDFDynamicHigherInclination} presents the CDFs of stretch factor and round trip time for inclination of $80^\circ$. Under the $80^\circ$ inclination, dynamic Starfield improves performance by \textbf{30\%} and static Starfield by \textbf{25\%} versus +Grid, reducing RTT by $26\,\text{ms}$. These improvements exceed those observed under the $53^\circ$ inclination, where dynamic and static Starfield achieve up to 15\% and 20\% reductions, respectively—an effect that is particularly pronounced for the dynamic variant (Section~\ref{evaluation:results:dynamic}). At the same time, we observe an overall degradation in stretch factor as the inclination increases. This can be attributed to the reduced satellite density resulting from covering an approximately 44\% larger surface area with the same number of satellites. The lower density increases inter-satellite distances, leading to longer links that are less effective at closely approximating geodesic paths, as anticipated by the analysis in Appendix~\ref{appendix:analysis}.

\subsubsection{Ablation Analysis}\label{evaluation:results:ablation} 
Building on the hop-count and stretch-factor trade-off introduced in Section~\ref{design:considerations:priority}, we conduct experiments for a shorter period under different traffic demands while varying the field-constant parameter ($K$) to evaluate Starfield’s performance in terms of hop count and stretch factor. Hop count is of particular interest because fewer hops generally lead to lower system processing delay, as packets traverse fewer satellites. Although per-satellite processing delays are typically sub-millisecond and often considered negligible, satellites are inherently power-constrained devices. Under conditions of inconsistent power availability—such as when operating on the Earth’s dark side~\cite{lalbakhsh2022darkening}—satellites may experience increased processing delays due to reduced computational capability. Moreover, hop count is directly correlated with link usage and thus has an indirect impact on queuing delay.
\begin{figure}[!t]
    \centering
    \includegraphics[width=\columnwidth]{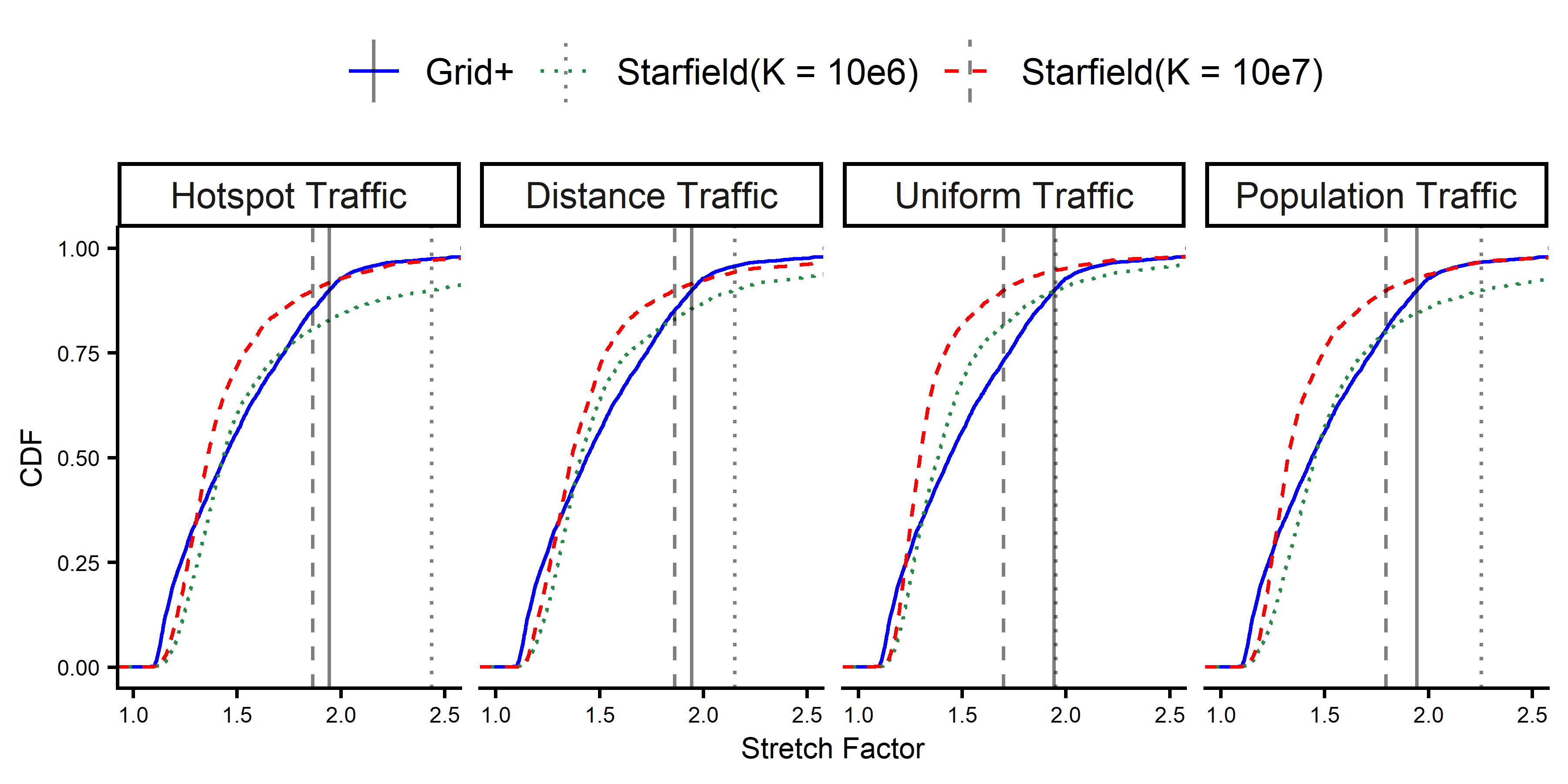}
    \caption{CDFs of city-to-city stretch factor for +Grid, Starfield with $K=10^6$ and $K=10^7$. Vertical lines indicate the 90th percentile.}
    \label{evaluation:results:ablation:StretchCDFAblation}
\end{figure}
\begin{figure}[!t]
    \centering
    \includegraphics[width=\columnwidth]{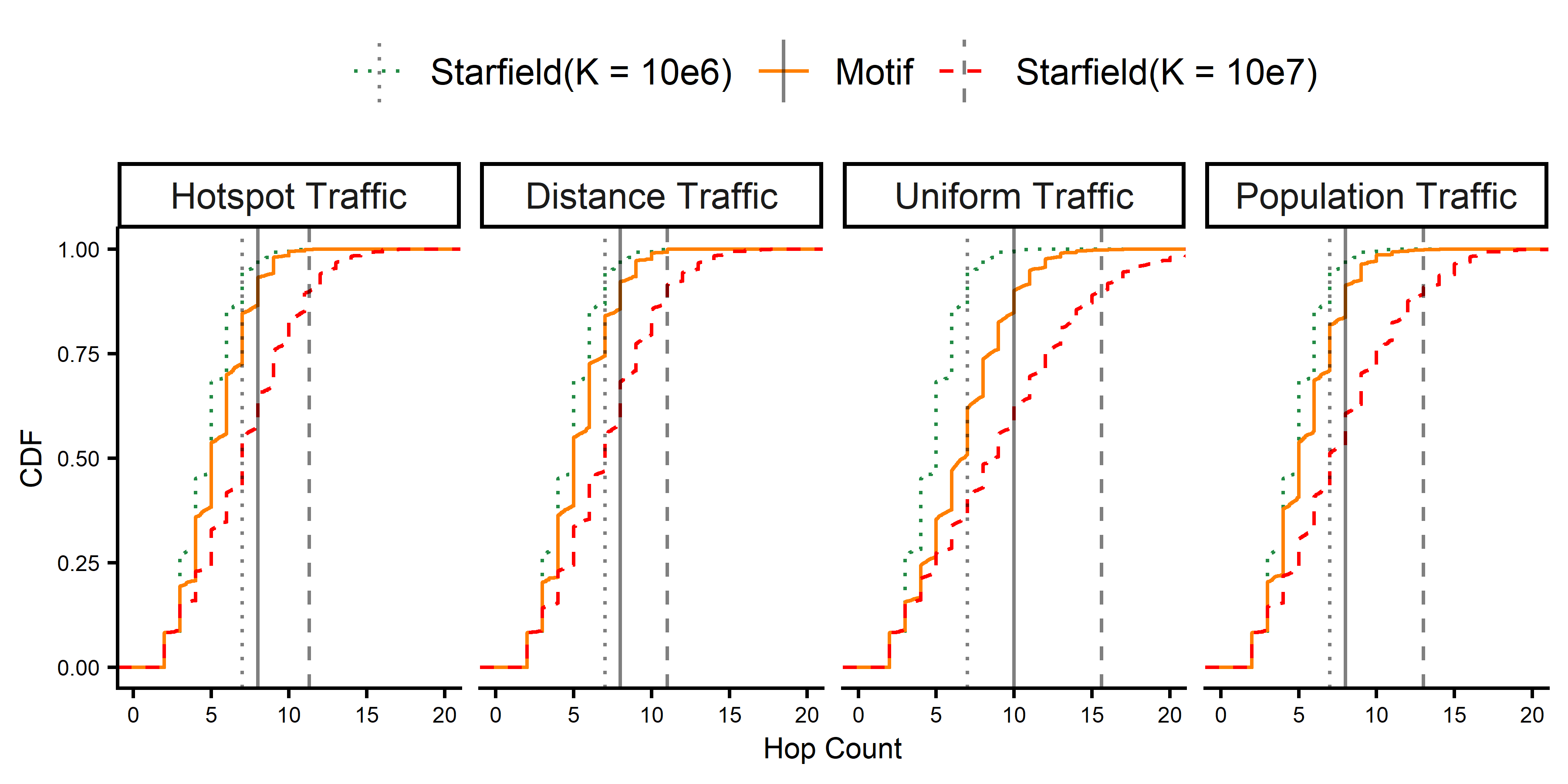}
    \caption{CDFs of city-to-city hop count for Motif, Starfield with $K=10^6$ and $K=10^7$. Vertical lines indicate the 90th percentile.}
    \label{evaluation:results:ablation:HopCountCDFAblation}
\end{figure}
Figure~\ref{evaluation:results:ablation:StretchCDFAblation} presents the CDFs of the stretch factor under different traffic demand patterns for Starfield with $K=10^6$, Starfield with $K=10^7$, and +Grid. Across all demand patterns, the $K=10^7$ case configuration consistently outperforms +Grid, whereas the $K=10^6$ case configuration underperforms, in accordance with the design principle that larger values of $K$ favor stretch-factor optimization. Figure~\ref{evaluation:results:ablation:HopCountCDFAblation} shows the CDFs of hop count for the same configurations. In contrast to stretch factor, Starfield with $K=10^6$, consistently achieves lower hop counts than $K=10^7$ configuration across all demand patterns against Motif—which overemphasizes hop-count optimization, aligning with the design principle that smaller values of $K$ favor hop-count minimization. Notably, both stretch and hop count gains are more pronounced under uniform demands rather than skewed traffic.

Moreover, to experimentally validate the contribution of our Riemannian formulation, we conduct an ablation on the metric (Table~\ref{evaluation:results:ablation:metric}). We replace our proposed metric calculation with a standard, baseline spherical metric ($g_0$). In this case the stretch factor degrades, increasing to the grid-like asymptotic limit of $\sqrt{2} \approx 1.42$. While this baseline performs better than other uniform topologies, it fails to outperform our demand-induced Riemannian metric ($g_0 + FF^T$) across any network metric under population demand. This resulting performance drop quantifies the critical importance of our proposed metric tensor over a uniform or demand-unaware geometric baseline.
\begin{table}[!htbp]
\centering
\begin{tabular}{lcccc}
\toprule
Performance Metric & Motif & +Grid & $g_0$ & $g_0 + FF^T$ \\
\midrule
Average stretch factor & 1.613 & 1.503 & 1.425 & \textbf{1.392} \\
Average hop count & 7.2 & 8.1 & 9.4 & \textbf{6.8} \\
Average RTT & \SI{121}{\milli\second} & \SI{109}{\milli\second} & \SI{101}{\milli\second} & \textbf{77 ms} \\
Average jitter & \SI{43.7}{\milli\second} & \SI{40.7}{\milli\second} & \SI{50.3}{\milli\second} & \textbf{11.8 ms} \\
Average queuing overhead & 1.44 & 1.36 & 1.51 & \textbf{0.99} \\
\bottomrule
\end{tabular}
\caption{Ablation on the Riemannian metric with demand-aware and spherical metric comparison}
\label{evaluation:results:ablation:metric}
\end{table}

\subsubsection{Starfield Robustness}\label{evaluation:results:robustness}
\begin{table}[!t]
\centering
\begin{tabular}{lcccc}
\toprule
Metric & Noiseless & $\mathcal{N}(0.5, 0.25)$ & $\mathcal{N}(1.0, 0.25)$ & $\mathcal{N}(1.0, 0.5)$ \\
\midrule
Average stretch factor & \textbf{1.396} & 1.397 & 1.4 & 1.4 \\
Average hop count & 10.4 & 10.3 & \textbf{10.1} & 10.1 \\
Average RTT & \SI{83.3}{\milli\second} & \SI{84}{\milli\second} & \textbf{\SI{82.4}{\milli\second}} & \SI{82.8}{\milli\second} \\
Average jitter & \SI{21.6}{\milli\second} & \SI{24.6}{\milli\second} & \textbf{\SI{19.1}{\milli\second}} & \SI{20.8}{\milli\second} \\
Packets routed & 1,558,174 & 1,982,315 & \textbf{2,498,065} & 2,488,551  \\
\bottomrule
\end{tabular}
\caption{Sensitivity Analysis of Starfield Performance Under Gaussian Demand Perturbations}
\label{evaluation:results:robustness:performance}
\end{table}

To assess the robustness of Starfield, we construct the topology using the demand pattern described in Section~\ref{evaluation:results:dynamic}. During simulation, we then introduce mild Gaussian perturbations to this demand pattern to evaluate how performance degrades under demand uncertainty. The experiments were conducted over an extended period, and the topology was constructed using the static Starfield design. Table~\ref{evaluation:results:robustness:performance} summarizes Starfield’s performance under different levels of Gaussian traffic noise, modeled as ($\mathcal{N}(\mu, \sigma)$): $\mathcal{N}(0.5, 0.25)$, $\mathcal{N}(1.0, 0.25)$, and $\mathcal{N}(1.0, 0.5)$. The injected noise corresponds to variations in the number of packets generated for each source–destination flow according to the specified distributions, resulting in noticeable modifications to the traffic pattern, as reflected by spikes in the total number of routed packets reported in the table. As observed, the mean ($\mu$) have a greater impact than the standard deviation ($\sigma$), while the degradation in stretch factor remains negligible (below 3\%). Furthermore, the demand dispersion introduced by the noise leads to slight improvements in hop count and RTT, indicating that Starfield remains resilient to moderate demand uncertainty.

\subsubsection{Visualizing Starfield}\label{evaluation:results:Visualization} 
\begin{figure}[!t]
  \centering
  \begin{minipage}[t]{0.48\linewidth}
    \centering
    \includegraphics[width=\linewidth]{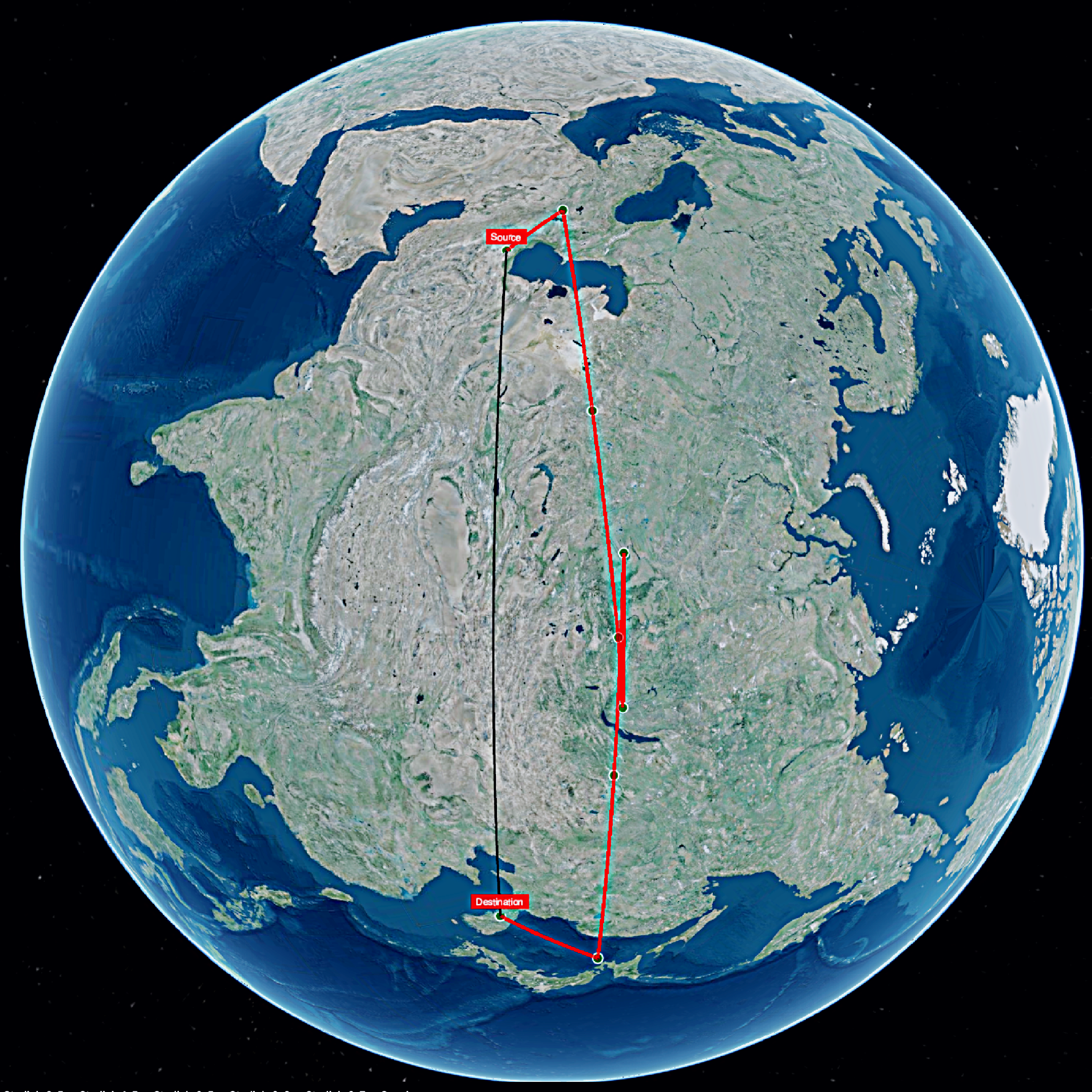}
  \end{minipage}\hfill
  \begin{minipage}[t]{0.48\linewidth}
    \centering
    \includegraphics[width=\linewidth]{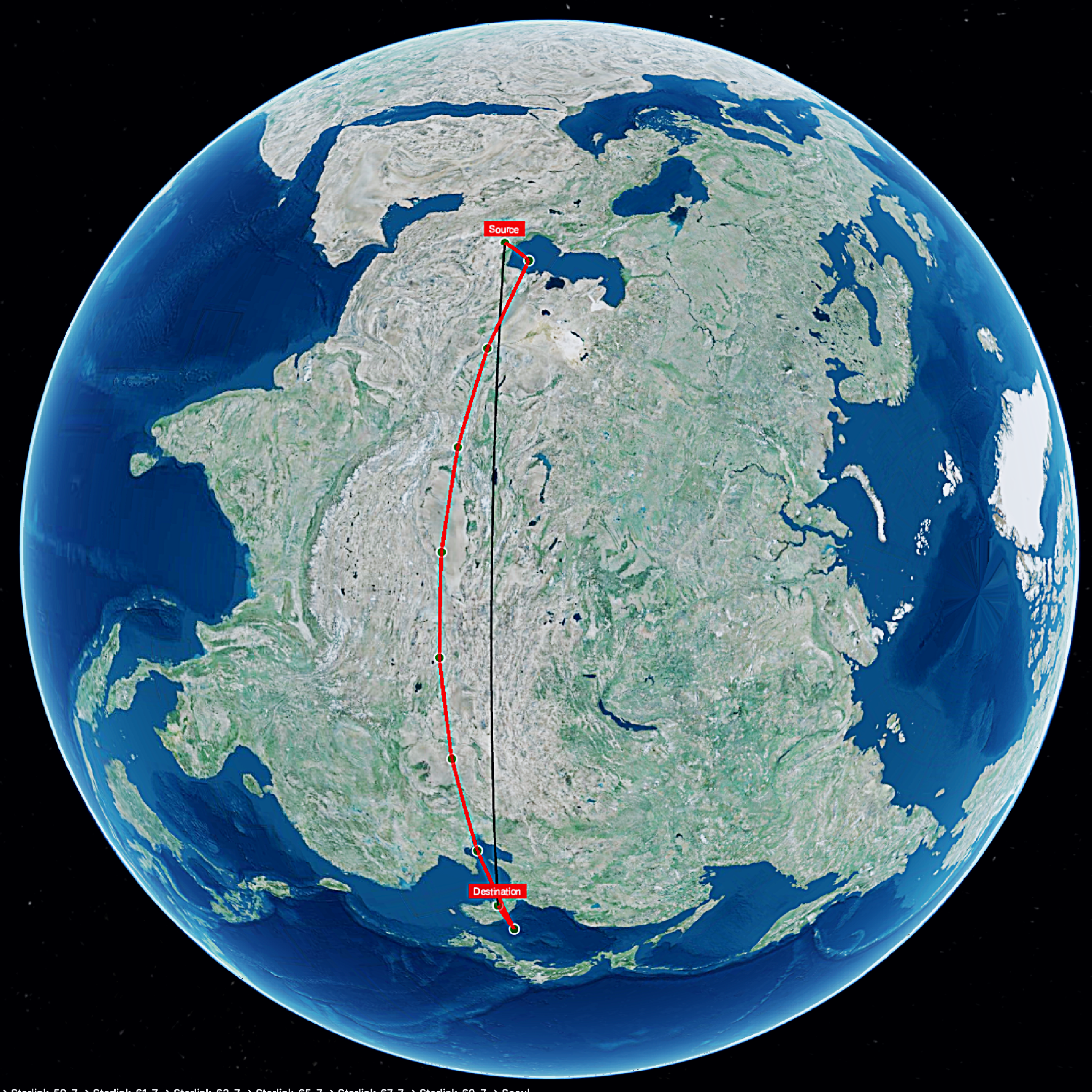}
  \end{minipage}
  \caption{The routed paths selected by +Grid (left) and Starfield (right). The black curve represents the geodesic, while the red path indicates the shortest routed path. The zig-zag routing behavior is clearly visible in the +Grid topology, whereas Starfield produces a smoother path that more closely follows the geodesic.}
  \label{evaluation:results:Visualization:path}
\end{figure}

\begin{figure}[!t]
  \centering
  \begin{minipage}[t]{0.48\linewidth}
    \centering
    \includegraphics[width=\linewidth]{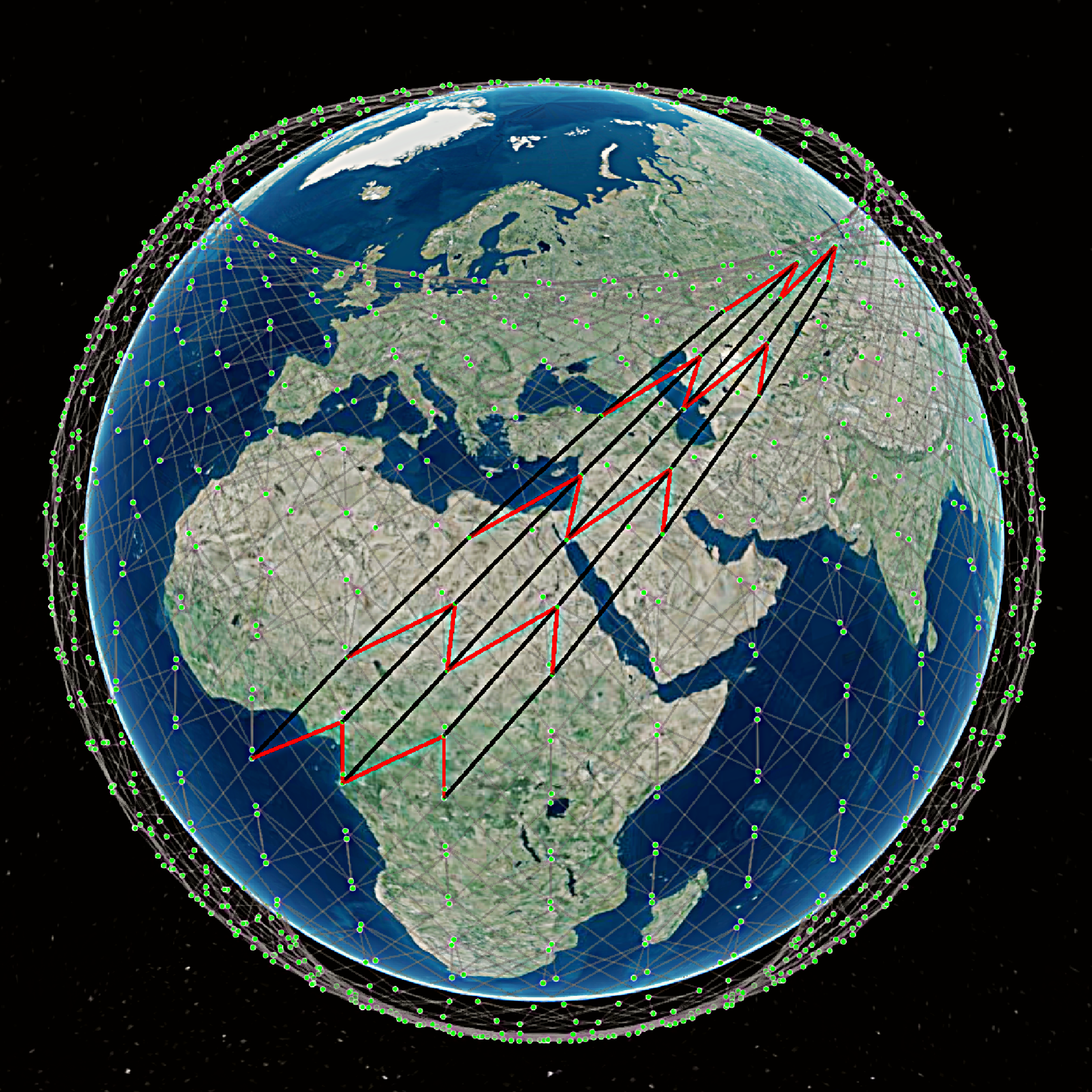}
  \end{minipage}\hfill
  \begin{minipage}[t]{0.48\linewidth}
    \centering
    \includegraphics[width=\linewidth]{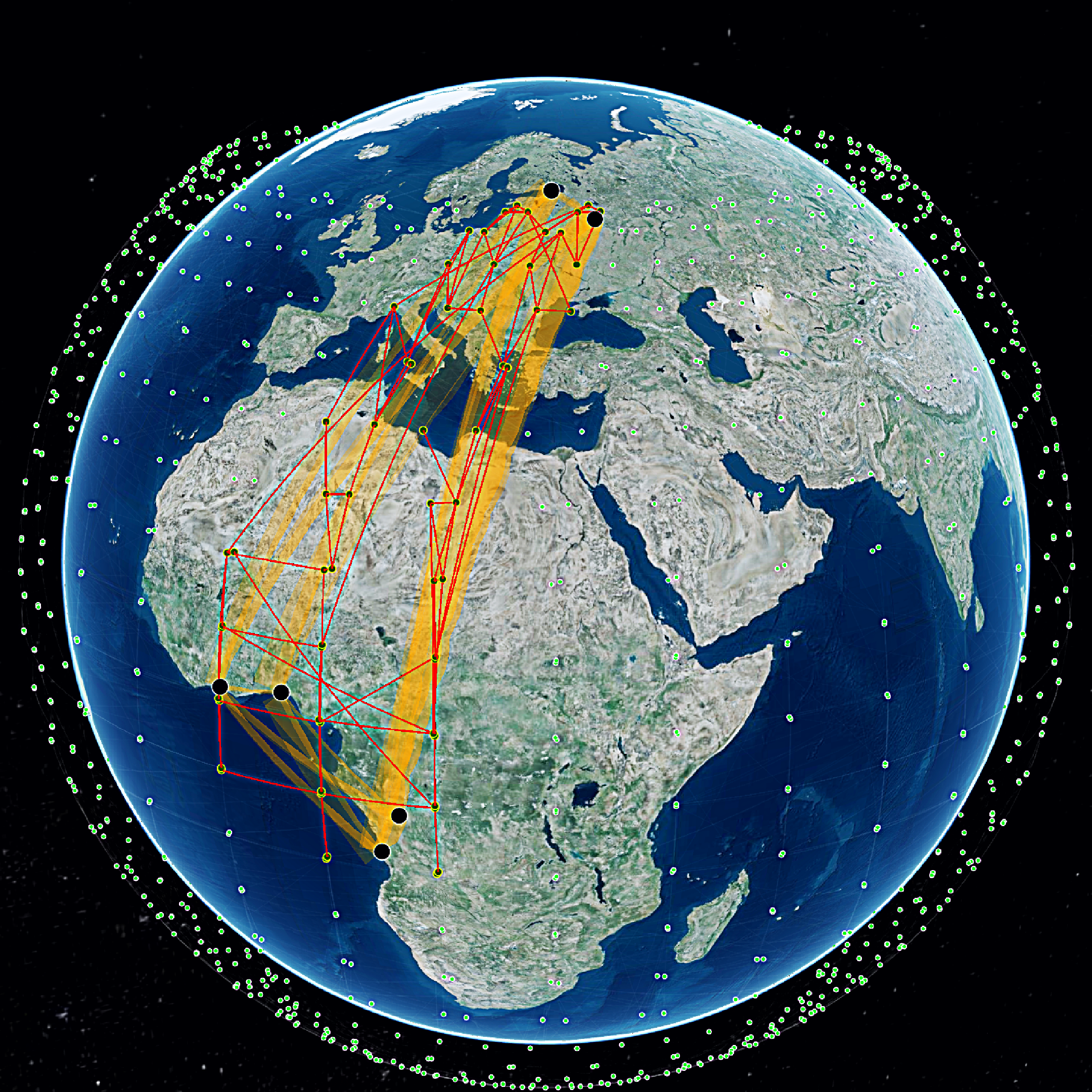}
  \end{minipage}
  \caption{+Grid (left) and Starfield (right) topologies. In +Grid (left), black lines denote intra-orbital ISLs and red lines represent inter-orbital ISLs. In the Starfield topology (right), red lines represent the selected ISLs, while orange lines depict the underlying traffic demands. The corresponding ground stations are shown as black dots. As shown, Starfield’s topology closely aligns with the traffic patterns.}
  \label{evaluation:results:Visualization:topology}
\end{figure}
To visually illustrate the performance advantages of Starfield, we visualize both the constellation and its network states using the Cesium visualizer~\cite{cesium}. Figure~\ref{evaluation:results:Visualization:path} compares the routing paths selected by Starfield (left) and +Grid (right) between a source and destination. As shown, Starfield produces a smoother path that more closely follows the geodesic curve. In contrast, +Grid exhibits a noticeable zig-zag subpath, which leads to a substantial degradation in stretch factor. Although Starfield incurs one additional hop in this example, this increase results from prioritizing path smoothness and closer alignment with the geodesic trajectory. Figure~\ref{evaluation:results:Visualization:topology} illustrates the topologies generated by Starfield (left) and +Grid (right). +Grid is constrained by the underlying orbital structure and, as a result, may fail to adequately capture the prevailing demand patterns. In contrast, Starfield adaptively orients inter-satellite links toward the dominant demand geodesics, leading to a topology whose link structure more closely aligns with the observed traffic demand trajectories.

\section{Related Work}\label{relatedWork}
Several prior studies have explored inter-satellite topology design with the goal of optimizing network performance, stability, reliability, utilization, or energy efficiency. Zhu et al.~\cite{zhu2022laser} select links based on inter-satellite visibility, while Yang et al.~\cite{yang2024analysis} propose a link selection weighting scheme derived from acquisition probability. Both approaches rely on the geometric properties of the constellation and primarily aim to construct a stable and reliable topology. Other works focus on reducing the number of inter-satellite links (ISLs). Wang et al.~\cite{wang2022intersatellite} prune ISLs from the +Grid topology to minimize bandwidth utilization, while Chen et al.~\cite{chen2025demand} introduce a zone-based dynamic topology control mechanism that improves energy efficiency and network reliability, albeit with graceful degradation in performance. A separate line of research emphasizes improving end-to-end network performance. McLaughlin et al.~\cite{mclaughlin2023grid} reduce hop count through a handoff-friendly topology design that leverages the positional predictability of satellites. Qiao et al.~\cite{qiao2023onboard} present an integer linear programming formulation that enables autonomous link selection by satellites through centralized coordination, thereby improving network delay and stability. Ren et al.~\cite{ren2025inter} propose a genetic algorithm–based topology design, where the elite selection criteria are the average and maximum network delay.

The primary work that also targets network performance optimization, and against which we conceptually compare Starfield—a greedy approach—is the motif-based search algorithm proposed by Bhattacherjee et al.~\cite{bhattacherjee2019network}. Similar to our study, their work identifies a trade-off between stretch factor and hop count (Section~\ref{evaluation:setup:metrics}). However, their results indicate bias toward hop count, and failure to outperform +Grid, introduced by Wood et al.~\cite{wood2001internetworking}, stretch factor. In contrast, Starfield incorporates a link-length prioritization mechanism that enables flexible trade-offs between hop count and stretch factor through the tuning of a single hyper-parameter (Section~\ref{design:considerations:priority} and Section~\ref{evaluation:results:ablation}). Starfield is explicitly demand-aware, accounting for geographically skewed traffic patterns (Section~\ref{motivation:demandGeometry}), rather than relying on uniform traffic assumptions, as is typical in motif-based, including +Grid, topology. Additionally, the motif-based approach constrains the topology search space to uniform motif structures, whereas Starfield does not impose such structural limitations. Finally, greedy algorithms are generally simpler, less exhaustive, and substantially faster than search algorithms, genetic algorithms, and linear programming. Consequently, Starfield offers superior time complexity compared to motif-based search.

The application of Riemannian geometry to networking is not new. The works~\cite{ni2015ricci, salamatian2022curvature} use the notion of Ricci curvature to analyze various path and connectivity characteristics of the Internet. Sarkar et al.~\cite{sarkar2009greedy} employ conformal mapping via Ricci flow to enable successful greedy forwarding in sensor networks. The emerging field of Riemannian metric learning has motivated several research works in wireless networks~\cite{gruffaz2025riemannian, qin2025pr}. Sadique et al.~\cite{sadique2025link} address dynamic link scheduling in LEO satellite networks using deep neural networks on Riemannian manifolds, while Shelim et al.~\cite{shelim2022geometric} apply geometric machine learning for link scheduling in device-to-device networks.

\section{Conclusion}\label{conclusion}
In this paper, we proposed \textbf{Starfield}, a demand-aware topology design that leverages a Riemannian metric induced by a vector field derived from traffic geometry and manifold boundary constraints. By orienting inter-satellite links toward geodesic paths, Starfield targets the minimization of end-to-end stretch factor, thereby reducing propagation delay. In addition, Starfield accounts for hop count and its trade-off with stretch factor through a tunable hyper-parameter, allowing the topology to be adapted to different performance objectives. To support the near-optimality of the proposed method, we provided a mathematical analysis. Starfield can be applied at any arbitrary period of the constellation’s evolution, and to promote static topological behavior, we introduced static Starfield based on inter-orbital matching. Using our custom-built packet-level simulator—which is more link-aware, link-configurable, and lightweight than existing state-of-the-art tools—we demonstrated that, in the Phase 1 Starlink constellation, Starfield achieves up to a 15\% reduction in stretch factor, and up to a 30\% reduction in hop count across diverse traffic patterns relative to state-of-the-art method LAMP. Static Starfield further achieves up to a 15\% improvement in stretch factor. We also evaluated the robustness of Starfield under traffic perturbations and observed less than a 3\% average degradation in stretch factor, indicating resilience to demand uncertainty. Finally, we visualized the resulting topologies and routing paths on a three-dimensional Earth model, illustrating Starfield’s ability to align network structure with underlying traffic demands.

As future work, Starfield can be extended to design topologies across multiple orbital shells, as next-generation constellations increasingly adopt multi-shell architectures at different altitudes. By prioritizing traffic demands and mapping higher-priority flows to lower-altitude shells—where shorter propagation distances are achievable—Starfield could further reduce latency. Realizing this capability would require equipping Starfield with a corresponding multi-shell routing protocol. In addition, Starfield can be extended to address non-stationary network dynamics, e.g., time-varying traffic demands and ISL establishment delays, as well as geometric effects arising from Earth’s rotation and the resulting variation in ground station orbital coverage over longer periods.
\bibliographystyle{ACM-Reference-Format}
\bibliography{bibliography}

\appendix

\section{Geometrical Calculations}\label{appendix:geomCalc}
\subsection{Vector Field Geometrical Components}\label{appendix:geomCalc:field}
To compute the geodesic on a sphere between two points $p$ and $q$, with position vectors $\vec{p}$ and $\vec{q}$, we first determine the azimuth angle $\alpha$ between them. The geodesic distance is then obtained by multiplying this angle by the shell radius $\rho$. The sine of the azimuth angle can be calculated using trigonometry in the triangle formed by the Earth’s center, $p$, and $q$. Accordingly, we derive both the azimuth angle and the geodesic distance as follows:
\begin{equation}\label{appendix:geomCalc:field:geodesic}
    \alpha = 2\sin^{-1}{\left(\frac{\|\vec{p} - \vec{q}\|}{2\rho}\right)} \Longrightarrow d^2_{\mathbf{S}^2}(p,q) = 2\rho\sin^{-1}{\left(\frac{\|\vec{p} - \vec{q}\|}{2\rho}\right)}
\end{equation}

To compute the unit tangent vector of the geodesic from point $p$ to $q$ at $p$, with position vectors $\vec{p}$ and $\vec{q}$, we consider the great circle $C$, defined as the intersection of the spherical shell with the plane passing through the Earth’s center, $p$, and $q$, since arc of $C$ connecting $p$ and $q$ constitutes the geodesic curve. By definition, the tangent vector at $p$ lies in the plane of $C$ and is orthogonal to $\vec{p}$, since $\vec{p}$ is the radius vector of $C$. Let $\vec{c^{\perp}}$ denote the normal vector of the great-circle plane, given by $\vec{c^{\perp}} = \vec{q} \times \vec{p}$. Because both $\vec{p}$ and and the tangent vector lie in the plane of $C$, the tangent vector is orthogonal to $\vec{c^{\perp}}$ is also perpendicular to the tangent. Consequently, the unit tangent vector can be obtained by taking the cross product of $\vec{c^{\perp}}$ and $p$, followed by normalization, yielding:
\begin{equation}\label{appendix:geomCalc:field:tangent}
    \vec{c^{\perp}} = \vec{q} \times \vec{p} \: , \: \vec{\tau}^{p}_{pq} = \vec{c^{\perp}} \times \vec{p}  \Longrightarrow \hat{\tau}^p_{pq} = \frac{(\vec{q} \times \vec{p}) \times \vec{p}}{\|(\vec{q} \times \vec{p}) \times \vec{p}\|}
\end{equation}

\subsection{Angular Link Selection}\label{appendix:geomCalc:angular}
For a satellite $s$ and a selected satellite $s^*$ with which we desire to establish a link, we choose a neighboring satellite $s'$ of $s$ such that the angle between links $ss'$ and $ss^*$ is closest to a target angle $\beta$. This is achieved by leveraging the inner product of the corresponding link vectors. Since the inner product yields the cosine of the angle between two vectors, and $\beta < \pi$, we compare this value to cosine ${\beta}$. Accordingly, the chosen neighbor is obtained by minimizing the following term: 
\begin{equation}\label{appendix:geomCalc:angular:minimumAngle}
    \zeta(\vec{s}, \vec{s^*}, \beta) = \min_{s'} \left|\frac{(\vec{s} - \vec{s'}) \cdot (\vec{s} - \vec{s^*})}{\|\vec{s} - \vec{s'}\|\|\vec{s} - \vec{s^*}\|} - \cos{\beta} \right|
\end{equation}
where $\vec{s}$, $\vec{s'}$, and $\vec{s^*}$ denote position vectors of $s$, $s'$, and $s^*$, respectively. However, to ensure fair link selection, each link must be chosen according to a consistent orientation across all satellites. For a given target link, there may exist two candidate neighbors on opposite sides with similar values derived from Equation~\ref{appendix:geomCalc:angular:minimumAngle}. To resolve this ambiguity, we enforce a consistent clockwise link orientation policy across satellites. This orientation is determined using the following criterion, applied prior to evaluating Equation~\ref{appendix:geomCalc:angular:minimumAngle}:
\begin{equation}\label{appendix:geomCalc:angular:orientationCriterion}
    sgn\left[\left((\vec{s} - \vec{s'}) \times (\vec{s} - \vec{s^*})\right) \cdot \vec{s} \right] > 0
\end{equation}

\subsection{Latitude and Longitude Unit Vectors}\label{appendix:geomCalc:latlon}
For a point $p$ with Cartesian coordinates $(X, Y, Z)$ on a sphere of radius $\rho$, we consider the circle of constant latitude $\theta$ as a curve $C_{\theta}$. By differentiating the parametrization of $C_{\theta}$ with respect to the longitude parameter $\varphi$, we obtain the longitude unit tangent vector at $p$, denoted by $\hat{\varphi}_p$, using the spherical-to-Cartesian coordinate transformation:
\begin{equation}\label{appendix:geomCalc:latlon:lon}
    C_{\theta} = 
    \begin{bmatrix}
    \sqrt{\rho^2-Z^2} \cos{\varphi} \\
    \sqrt{\rho^2-Z^2} \sin{\varphi} \\
    Z \\
    \end{bmatrix} \Longrightarrow
    \frac{\frac{\partial C_{\theta}}{\partial \varphi}}{\|\frac{\partial C_{\theta}}{\partial \varphi}\|} = 
    \begin{bmatrix}
    - \sin{\varphi} \\
    \cos{\varphi} \\
    0 \\
    \end{bmatrix} \Longrightarrow \hat{\varphi_p} = 
    \begin{bmatrix}
    \frac{-Y}{\sqrt{X^2+Y^2}} \\
    \frac{X}{\sqrt{X^2+Y^2}} \\
    0 \\
    \end{bmatrix}
\end{equation}
Similarly, we consider the curve of constant longitude $\varphi$ as a curve $C_{\varphi}$. By differentiating the parametrization of $C_{\varphi}$ with respect to the latitude parameter $\theta$, we obtain the latitude unit tangent vector at $p$, denoted by $\hat{\theta}_p$, using the spherical-to-Cartesian coordinate transformation:
\begin{equation}\label{appendix:geomCalc:latlon:lat}
    C_{\varphi} = 
    \begin{bmatrix}
    \rho \frac{X}{\sqrt{X^2+Y^2}} \cos{\theta} \\
    \rho \frac{Y}{\sqrt{X^2+Y^2}} \cos{\theta} \\
    \rho \sin{\theta} \\
    \end{bmatrix} \Longrightarrow
    \frac{\frac{\partial C_{\varphi}}{\partial \theta}}{\|\frac{\partial C_{\varphi}}{\partial \theta}\|} = 
    \begin{bmatrix}
    - \frac{X}{\sqrt{X^2+Y^2}} \sin{\theta} \\
    - \frac{Y}{\sqrt{X^2+Y^2}} \sin{\theta} \\
    \cos{\theta} \\
    \end{bmatrix} \Longrightarrow \hat{\theta} = 
    \begin{bmatrix}
    \frac{-XZ}{\rho \sqrt{X^2+Y^2}} \\
    \frac{-YZ}{\rho \sqrt{X^2+Y^2}} \\
    \frac{\sqrt{X^2+Y^2}}{\rho} \\
    \end{bmatrix}
\end{equation}

\subsection{Region Flows}\label{appendix:geomCalc:geodesicRegion}
To determine whether a geodesic flow—defined by source and destination position vectors $\vec{u}$ and $\vec{v}$, respectively—intersects a given region, we first define a representative point $p$ for the region. $p$ chosen as the intersection of the region’s midpoint latitude and longitude lines. We then project $p$ onto the geodesic curve connecting $\vec{u}$ and $\vec{v}$ to obtain the closest point on the curve, denoted by $q$. Let $\vec{p}$ and $\vec{q}$ denote the position vectors of $p$ and $q$, respectively. We compute $\vec{q}$ as follows:
\begin{equation}\label{appendix:geomCalc:geodesicRegion:projection}
    \hat{n} = \frac{\vec{u} \times \vec{v}}{\|\vec{u} \times \vec{v}\|} \Longrightarrow \vec{q} = \rho \frac{\vec{p}-(\vec{p} \cdot \hat{n})\hat{n}}{\|\vec{p}-(\vec{p} \cdot \hat{n})\hat{n}\|}
\end{equation}
where $\rho$ denotes the spherical radius, and $\hat{n}$ is the normal vector of the plane defined by $u$, $v$, and the and the Earth’s center. We then determine the region identifier to which $q$ belongs. Specifically, the latitude and longitude indices of $q$ are $\left\lfloor \frac{\varphi_q}{l_\varphi} \right\rfloor$ and $\left\lfloor \frac{\theta_q+\frac{\pi}{2}}{l_\theta} \right\rfloor$, respectively, where $\varphi_q$, and $\theta_q$ denote the longitude and latitude of $q$, and $l_\varphi$ and $l_\theta$ the uniform angular step sizes in longitude and latitude. If these indices match those of the region, we conclude that the flow passes through the region.

\section{Mathematical Analysis}\label{appendix:analysis}
\subsection{Model}
We assume the manifold is a 2-dimensional plane with the usual Euclidean metric. 
We denote the satellites as points (nodes) on the plane, with $V$ being the set of all satellites. 
For a satellite $u \in V$, $\mathbf{x}_v = (x_v^0, x_v^1)$ denotes the Cartesian coordinates of the satellite on the plane. 
We assume the satellites are stationary for the analysis; operationally, our model captures the properties of the satellite network over a small time window (\~ a few ms) when the satellites are relatively stationary with respect to each other. 
A satellite has a maximum range of $R > 0$, i.e., it can connect with another satellite so long as the Euclidean distance between the satellites is less than $R$. 
We assume satellites have a degree bound of $\delta$; in practice $\delta = 4$ for the Starlink system. 
Let $E$ denote the set of all edges (connections) between the satellites. 
Satellites are homogeneously distributed over the plane, with a density of $\rho > 0$ and a distortion of $0 < \eta < 1/(2\rho)$ defined as follows:  
a circle of radius $\eta$ on the plane centered at any point $(i/\rho, j/\rho)$ where $i, j \in \mathbb{Z}$ contains exactly one satellite. 
If satellites are arranged in an approximate 2-dimensional square grid patter, the $\rho$ parameter captures the interval length between adjacent satellites on the grid, while the $\eta$ parameter captures by how much the satellite locations deviate from being a true grid. 
For a fixed $\rho, \eta$, our definition implies that a satellite is guaranteed to have another satellite within a distance of $1/\rho + 2\eta $ from it along each of the $x$, negative $x$, $y$ and negative $y$ axis directions on the plane. 

We denote the set of demands by $\mathcal{D} = \{ (s, d): s, d \in V, s \neq d \}$. 
If $(s, d) \in \mathcal{D}$ we say there is a demand from source $s$ to destination $d$. 
We do not associate any bandwidth with the demands, and assume the system has adequate capacity to handle the demands. 
A path is a sequence of connected satellites. 
For any path $u_0, u_1, \ldots, u_k$ with $u_i \in V$ for all $i = 0, 1, \ldots, k$, the latency of sending a packet from $u_0$ to $u_k$ is given by $\sum_{i=0}^{k-1} \| \mathbf{x}_{u_i} - \mathbf{x}_{u_{i+1}} \|_2$.  
The geodesic distance between the source $u_0$ and the destination $u_k$ of the path is $\| \mathbf{x}_{u_0} - \mathbf{x}_{u_k} \|_2$. 
The stretch of the path is the ratio $\sum_{i=0}^{k-1} \frac{\| \mathbf{x}_{u_i} - \mathbf{x}_{u_{i+1}} \|_2}{\| \mathbf{x}_{u_0} - \mathbf{x}_{u_k} \|_2}$ of the path length to the length of the geodesic path. 

For any demand $(s, d) \in \mathcal{D}$, we assume the packets from $s$ to $d$ are routed via the shortest path from $s$ to $d$ on the topology $E$ (with ties broken arbitrarily).
The weight of an edge in $E$ is equal to the Euclidean distance between the endpoints of the edge in the shortest path computations. 
We denote by $\sigma_E(s, d)$ the stretch of the shortest path from $s$ to $d$ on the topology $E$.  
The primary metric of interest is the worst case stretch $\max_{(s, d) \in \mathcal{D}} \sigma_E(s, d)$. 

\subsection{Converse}

We show a lower bound on the stretch under the following assumption: for any $(s, d) \in \mathcal{D}$, the shortest path from $s$ to $d$ on the topology $E$ is contained within the  region $P_{(s,d)} = \{\mathbf{x} \in \mathbb{R}^2: \exists~ \mathbf{y} \text{ on the straight line between } s \text{ and } d \text{ such that } \| \mathbf{x} - \mathbf{y} \|_2 \leq R \} $.

Consider any continuous, compact (closed, bounded) region $\mathcal{A} \subset \mathbf{R} \times \mathbf{R}$. 
Let $\mathcal{D}_\mathcal{A} \subseteq \mathcal{D}$ be the maximal set of demands contained in $\mathcal{A}$, i.e., (1) for any $(s, d) \in \mathcal{D}_\mathcal{A}$ the region $P_{(s,d)} \subseteq \mathcal{A}$, and (2) for any $(s, d) \in \mathcal{D} \backslash \mathcal{D}_\mathcal{A}$ the region $P_{(s,d)} \nsubseteq \mathcal{A}$. 
For each demand $(s, d) \in \mathcal{D}_\mathcal{A}$, let $ l(s,d) = \| \mathbf{x}_{s} - \mathbf{x}_{d}\|_2 $ be the length of the demand. 

Let $\mathrm{cell}(i,j)$ for $i, j \in \mathbb{Z}$ denote the region $[i/\rho, (i+1)/\rho] \times [j/\rho, (j+1)/\rho]$. 
Let $q(\mathcal{A}) = \cup_{i, j \in \mathbb{Z}: \mathrm{cell}(i,j) \cap \mathcal{A} \neq \{ \}} \mathrm{cell}(i,j)$ denote a discretization of $\mathcal{A}$ that includes all cells that have a nonempty intersection with $\mathcal{A}$.  
Let $\alpha(\mathcal{A})$ denote the area of $q(\mathcal{A})$. 
There are $\alpha(\mathcal{A}) \rho^2$ cells in $q(\mathcal{A})$ and at most $2\alpha(\mathcal{A})\rho^2 + 2$ satellites in $q(\mathcal{A})$ and hence in $\mathcal{A}$.  
This implies there are at most $(2\alpha(\mathcal{A}) \rho^2 + 2) \delta$ edges in the region $\mathcal{A}$.
Here, we say an edge $(u, v) \in E$ is in the region $\mathcal{A}$ if the line segment joining $\mathbf{x}_u$ and $\mathbf{x}_v$ has a non-zero intersection with $\mathcal{A}$.  
Let $E_\mathcal{A} \subseteq E$ denote the edges in $\mathcal{A}$. 
Each edge $(u, v) \in E_\mathcal{A}$ has an associated orientation angle $\theta(u ,v) \in [0, 2\pi)$ defined as the angle between the vector $\mathbf{x}_v - \mathbf{x}_u$ and the $x$-axis, with the angle measured in an anti-clockwise direction starting from the $x$-axis.  
We similarly define orientation angle $\theta(s,d)$ for each demand $(s,d) \in \mathcal{D}_\mathcal{A}$ as the angle of the vector $\mathbf{x}_d - \mathbf{x}_s$. 
Each edge $(u, v) \in E$ also has an associated length $\| \mathbf{x}_u - \mathbf{x}_v\|_2$ of at most $R$. 

For a demand $(s, d)$ with an orientation of $\theta(s, d)$, a lower bound on stretch can be given as follows.  
Consider a parameter $\epsilon \in (0, \pi)$. 
For any edge $(u, v) \in E_\mathcal{A}$, let $\mathcal{D}_{\mathcal{A}}^{(u, v)}(\epsilon) \subseteq \mathcal{D}_{\mathcal{A}}$ be the demands oriented at an angle that is at most $\epsilon$ from $(u, v)$. 
For any demand $(s, d) \in \mathcal{D}_{\mathcal{A}}$, let $E_\mathcal{A}^{(s,d)}(\epsilon)$ be the edges that are at an angle of at most $\epsilon$ from $(s, d)$. 
If an edge $(u, v) \in E_\mathcal{A}^{(s,d)}(\epsilon)$, we say that $(u,v)$ is $\epsilon$-near $(s,d)$. 
If a demand $(s,d)$ has sufficiently many $\epsilon$-near edges, the stretch for the demand can be low (though it is not guaranteed).
However, if a demand $(s, d)$ does not have adequately many $\epsilon$-near edges, the stretch for the demand must necessarily be high. 
We formalize this intuition below. 

\begin{theorem}
Consider any topology $E$. 
For any continuous, compact region $\mathcal{A} \subset \mathbb{R} \times \mathbb{R}$ with a non-empty $\mathcal{D}_\mathcal{A}$ and parameter $\epsilon \in [0, \pi)$, there exists a demand $(s,d) \in \mathcal{D}_\mathcal{A}$ with a stretch 
\begin{align}
\sigma_E(s,d) \geq \frac{ \min \left(l(s,d), \frac{R \lambda_\mathcal{A} (2\alpha(\mathcal{A}) \rho^2 + 1) \delta}{ |\mathcal{D}_\mathcal{A}|} \right)  + \frac{l(s,d) -   \min \left(l(s,d), \frac{R \lambda_\mathcal{A} (2\alpha(\mathcal{A}) \rho^2 + 1) \delta}{ |\mathcal{D}_\mathcal{A}|} \right)}{| \cos(\epsilon) |} }{ l(s,d) },
\end{align} 
where $\lambda_\mathcal{A}$ is the maximum number of demands in $\mathcal{D}_\mathcal{A}$ that are oriented at an angle within $\epsilon$ of any reference direction.   
\end{theorem}
\begin{proof}
The minimum number of edges needed in $E_\mathcal{A}^{(s,d)}(\epsilon)$  to satisfy a demand $(s, d)$ is $\frac{l(s,d)}{ R }$.
The minimum occurs if all the edges in $E_\mathcal{A}^{(s,d)}(\epsilon)$ are aligned precisely with $(s, d)$.
If $|E_\mathcal{A}^{(s,d)}(\epsilon)| < \frac{l(s,d)}{R}$, it implies there is at least one edge in the $s$ to $d$ shortest path who orientation angle differs from $\theta(s,d)$ by more than $\epsilon$. 
If $|E_\mathcal{A}^{(s,d)}(\epsilon)| = 0$, it implies the stretch of $(s, d)$ is at least $\frac{1}{ |\cos(\epsilon)| }$.  
Note that 
\begin{align}
\sum_{(u, v) \in E_\mathcal{A}} | \mathcal{D}_\mathcal{A}^{(u,v)}(\epsilon) | = \sum_{(s, d) \in \hat{\mathcal{D}}_\mathcal{A}} | E_\mathcal{A}^{(s, d)}(\epsilon)|.
\end{align}   

Suppose the demands are such that for any orientation of an edge $(u, v)$ we have $|\mathcal{D}_\mathcal{A}^{(u,v)}(\epsilon)|$ is less than $\lambda_\mathcal{A}$. 
Then 
\begin{align}
\sum_{(u, v) \in E_\mathcal{A}} | \mathcal{D}_\mathcal{A}^{(u,v)}(\epsilon) | \leq \lambda_\mathcal{A} |E_\mathcal{A}| = \lambda_\mathcal{A} (2\alpha(\mathcal{A}) \rho^2 + 2) \delta.
\end{align} 
This implies  $\sum_{(s, d) \in \mathcal{D}_\mathcal{A}} | E_\mathcal{A}^{(s, d)}(\epsilon)| \leq \lambda_\mathcal{A} (2\alpha(\mathcal{A}) \rho^2 + 2) \delta$.  
Therefore, there exists at least one demand $(s, d) \in \mathcal{D}_\mathcal{A}$ such that $|E_\mathcal{A}^{(s,d)}(\epsilon)| \leq \frac{\lambda_\mathcal{A} (2\alpha(\mathcal{A}) \rho^2 + 2) \delta}{ |\mathcal{D}_\mathcal{A}| }$.
To lower bound the stretch for this demand, we consider the best case scenario in which all $\frac{\lambda_\mathcal{A} (2\alpha(\mathcal{A}) \rho^2 + 2) \delta}{ |\mathcal{D}_\mathcal{A}| }$ edges are aligned in the exact same direction as $(s,d)$ and the remaining edges on the path from $s$ to $d$ are aligned at an angle exactly $\epsilon$ away from $(s,d)$.  
The total length of the edges that are aligned along $(s,d)$ is $ \min \left(l(s,d), \frac{R \lambda_\mathcal{A} (2\alpha(\mathcal{A}) \rho^2 + 2) \delta}{ |\mathcal{D}_\mathcal{A}|} \right) $.  
The total length of the edges that are aligned at an angle $\epsilon$ away from $(s,d)$ is $\frac{l(s,d) -   \min \left(l(s,d), \frac{R \lambda_\mathcal{A} (2 \alpha(\mathcal{A}) \rho^2 + 2) \delta}{ |\mathcal{D}_\mathcal{A}|} \right)}{|\cos(\epsilon)|}$.  
This gives a lower bound on the stretch as mentioned in the theorem.  
\end{proof} 

\subsection{Achievability}
To enable tractability of analysis, we consider a slight variation of the Starfield algorithm (denoted by $\mathrm{ALG}$) as explained below. 
Consider a set of demands $\mathcal{D}$. 
Let $E_\mathrm{ALG}$ be the set of edges constructed by $\mathrm{ALG}$. 
We are interested in deriving an upper bound for the stretch $\sigma_{E_\mathrm{ALG}}(s,d)$ of each demand $(s,d) \in \mathcal{D}$. 
For a demand $(s,d) \in \mathcal{D}$, if there are no other heavy demands near the straight line from $s$ to $d$, we show that the upper bound $\sigma_{E_\mathrm{ALG}}(s,d)$ is close to 1. 

Let $\tau \in  (1/\rho) \mathbb{N}$ be a tessellation parameter. 
We partition the plane into approximate square regions of side length $\tau$ as follows. 
For $i, j \in \mathbb{Z}$, we denote by $R_{i, j}$ the $(i, j)$-th region and by $\mathcal{R} = \{ R_{i,j}: i, j \in \mathbb{Z} \}$ the set of all regions in the plane. 
The $(i,j)$-th region $R_{i,j}$ includes each satellite $v \in V$ that is closest to a point in $\{ i\tau, i\tau + 1/\rho, i\tau + 2/\rho, \ldots, i\tau + (\tau - 1) / \rho \} \times \{ j\tau, j\tau + 1/\rho, j\tau + 2/\rho, \ldots, j\tau + (\tau - 1) / \rho \} $. 
Satellites in $R_{i,j}$ that are at the boundary of the region, i.e., closest to the points $(i\tau, j\tau), (i\tau, j\tau + 1/\rho), \ldots, (i\tau, j\tau + (\tau\rho - 1)/\rho)$ or $(i\tau + (\tau\rho-1)/\rho, j \tau), (i\tau + (\tau\rho-1)/\rho, j \tau + 1/\rho), \ldots, (i\tau + (\tau\rho-1)/\rho, j \tau + (\tau\rho-1)/\rho)$ or $(i\tau, j\tau), (i\tau + 1/\rho, j\tau), \ldots, (i\tau + (\tau\rho-1)/\rho, j\tau)$ or $(i\tau, j\tau + (\tau\rho-1)/\rho), (i\tau + 1/\rho, j\tau + (\tau\rho-1)/\rho), \ldots, (i\tau + (\tau\rho-1)/\rho, j\tau + (\tau\rho-1)/rho)$ are called boundary satellites of the region $R_{i,j}$.    
We let $\partial R_{i,j}$ denote the boundary satellites of $R_{i,j}$. 

The topology of $\mathrm{ALG}$ is constructed as follows. 
For each region $R_{i,j}$, we compute a primary direction $\mathbf{p}(R_{i,j})$ and a secondary direction $\mathbf{s}(R_{i,j})$ with $\mathbf{p}(R_{i,j}) \neq \mathbf{s}(R_{i, j})$ or  $\mathbf{p}(R_{i,j}) \neq -\mathbf{s}(R_{i,j})$. 
The primary and the secondary directions are computed from the Riemann metric at the region as follows.
At any point $\mathbf{x} \in \mathbb{R}^2$, let $g_\mathbf{x} \in \mathbb{S}_{++}^2$ be the metric where $\mathbb{S}_{++}^2$ denotes the set of $2\times 2$ positive-definite matrices. 
Here, $g_\mathbf{x}(i, j)$ for $i=0,1$ and $j=0,1$ is the inner product of the standard normal vectors along dimensions $i$ and $j$ respectively. 
Let $\mathbf{e}_\mathbf{x}$ be the eigenvector corresponding to the minimum eigenvalue of $g_\mathbf{x}$. 
Since $\mathbf{e}_\mathbf{x}$ and $-\mathbf{e}_\mathbf{x}$ are both eigenvectors corresponding to the minimum eigenvalue, for consistency we choose $\mathbf{e}_\mathbf{x}$ as the vector having a positive Euclidean inner product with the standard normal vector along the $x$-axis. 
If the inner product is zero, we choose the vector having a positive Euclidean inner product with the standard normal vector along the $y$-axis. 
The primary direction is defined as 
\begin{align}
\mathbf{p}(R_{i,j}) = \frac{ \sum_{\substack{v \in R_{i,j} \\ v \neq \partial R_{i,j}}} \mathbf{e}_{\mathbf{x}_v}}{ |\{v \in R_{i,j}: v \neq \partial R_{i,j} \}|}. 
\end{align}
We define $\mathbf{s}(R_{i,j}) $ as the vector orthogonal to $\mathbf{p}(R_{i,j})$, in a direction $90^{\circ}$ anti-clockwise to $\mathbf{p}(R_{i,j})$.  

Consider the continuous region $Q_{i,j} = [i\tau, i\tau + (\tau\rho - 1)/\rho] \times [j\tau, j\tau + (\tau\rho - 1)/\rho]$. 
As before, we call the points on the periphery of this grid as boundary points $\partial Q_{i,j}$ of $Q_{i,j}$. 
We draw a square grid with a line-spacing of $R/2$ and with the grid rotated such that a grid line is parallel to $\mathbf{p}(R_{i,j})$ on this region. 
Let $\bar{Q}_{i,j} \subset Q_{i,j}$ be the set of points where either (1) two grid lines intersect with each other, or (2) a grid line intersects with $\partial Q_{i,j}$. 
For any $\mathbf{x} \in \bar{Q}_{i,j}$ we denote its neighbors by $\Gamma_\mathbf{x} \subset \bar{Q}_{i,j}$, which are the intersection points that are adjacent to $\mathbf{x}$ on the grid or on the boundary. 
Each point $\mathbf{x} \in \bar{Q}_{i,j}$ that is in the interior, i.e., $\mathbf{x} \notin \partial Q_{i,j}$, has 4 neighbors in $\bar{Q}_{i,j}$.  
A boundary intersection point $\mathbf{x} \in \partial Q_{i,j} \cap \bar{Q}_{i,j}$ has either 2 or 1 neighbor in $\bar{Q}_{i,j}$.  
For any $\mathbf{x} \in \bar{Q}_{i,j}$, let $v_\mathbf{x} \in R_{i,j}$ denote the satellite that is the closest to the point $\mathbf{x}$. 
To construct the topology for each $\mathbf{x} \in \bar{Q}_{i,j}$ we connect $v_\mathbf{x}$ to $v_\mathbf{y}$ where $\mathbf{y} \in \Gamma_\mathbf{x}$. 
For any $\mathbf{x} \in \bar{Q}_{i,j}$, note that the distance between $\mathbf{x}$ and $\mathbf{x}_{v_\mathbf{x}}$ is at most $2\eta + 1/(\sqrt{2}\rho)$. 
The distance between adjacent intersection points on the grid is $R/2$. 
Therefore, the distance between any two connected satellites is at most $R/2 + 2(2\eta + 1/(\sqrt{2}\rho))$ which is less than the maximum range of $R$ if the density $\rho$ is sufficiently large. 
Furthermore, for a sufficiently large $\rho$ each intersection point $\mathbf{x} \in \bar{Q}_{i,j}$ has a unique closest satellite $v_\mathbf{x}$. 
This ensures the degree of each satellite is at most 4. 

We repeat the above process for all regions in the plane.
A technical difficulty arises in this construction. 
Consider adjacent regions $Q_{i,j}$ and $Q_{i+1,j}$ with a shared boundary line. 
If $\mathbf{p}(R_{i,j}) \neq \mathbf{p}(R_{i+1,j})$, the number of intersection points in $\bar{Q}_{i,j}$ on the shared boundary can be different from the number of intersection points in $\bar{Q}_{i+1,j}$ on the shared boundary. 
To ensure connectivity between regions, we connect the boundary points of the $(i,j)$ region to the boundary points of the $(i+1, j)$ region in the following way. 
For each point $\mathbf{x}$ in the shared boundary, let $\mathbf{y}_\mathbf{x}^{(i,j)} \in \bar{Q}_{i,j}$ and $\mathbf{y}_\mathbf{x}^{(i+1,j)} \in \bar{Q}_{i+1,j}$ be the boundary points that are the closest to $\mathbf{x}$ from regions $(i,j)$ and $(i+1, j)$ respectively. 
For each $\mathbf{x}$, we connect $v_{\mathbf{y}_\mathbf{x}^{(i,j)}}$ with $v_{\mathbf{y}_\mathbf{x}^{(i+1,j)}}$.
The above connections are always possible to make, since the number of intersection points on the shared boundary for the two regions differ by a factor of at most $\sqrt{2}$. 
The length of the newly added connections is at most $R/2$. 
 
Another technical difficulty that arises is that it is possible for only a subset of the satellites to be connected in the above construction. 
When choosing a large grid spacing ($R/2$, which is  $\gg 1/\rho$), the density of intersection points on the grid is smaller than $\rho$. 
This forces some of the satellites to be omitted from the topology. 
However, if we consider a denser grid with a smaller spacing that is comparable to $1/\rho$, rotating the grid becomes difficult, i.e., the edge lengths on a rotated grid can become a constant factor larger than the desired length defeating the purpose of grid rotation.
A solution to this technical issue is to construct multiple grids (each grid offset by a small distance) each containing a disjoint set of satellites. 
Assuming a ground station has sufficient range to talk to nearby satellites on any of the multiple grids, routing can be done on a per-grid basis. 
Therefore, in the following we ignore the possible presence of multiple grids, and focus our analysis on just a single grid per region.    

Consider a demand $(s,d) \in \mathcal{D}$.
We upper bound the length $l(p_{s,d})$ of the shortest path $p_{s,d}$ from $s$ to $d$ on our constructed topology as follows. 
Consider the straight line $\gamma:[0,1] \rightarrow \mathbb{R}^2$ joining $s$ to $d$ with $\gamma(0) = \mathbf{x}_s$ and $\gamma(1) = \mathbf{x}_d$. 
Let this line pass through the regions $(i_1, j_1), (i_2, j_2), \ldots, (i_k, j_k)$ in that order, with $\mathbf{x}_s \in Q_{i_1, j_1}$ and $\mathbf{x}_d \in Q_{i_k, j_k}$. 
Each region $h \in \{1, 2, \ldots, k\}$ has an entry point $\mathbf{x}^h_\mathrm{in}$ and an exit point $\mathbf{x}^h_\mathrm{out}$ on the line, i.e., if $\gamma(t^h_\mathrm{in}) = \mathbf{x}^h_\mathrm{in}$ and $\gamma(t^h_\mathrm{out}) = \mathbf{x}^h_\mathrm{out}$, then none of the points in $\{ \gamma(t): 0 \leq t \leq t^h_\mathrm{in} \text{ or } t^h_\mathrm{out} \leq t \leq 1 \}$ belong to $Q_{i_h, j_h}$.  
For the first region we have $\mathbf{x}_s = \mathbf{x}^1_\mathrm{in}$, and $\mathbf{x}_d = \mathbf{x}^k_\mathrm{out}$ for the last. 
Next, for any $h \in \{1, 2, \ldots, k\}$ let $\mathbf{y}^h_\mathrm{in} \in \bar{Q}_{i_h, j_h}$ be the intersection point that is the closest to $\mathbf{x}^h_\mathrm{in}$ in region $(i_h, j_h)$.  
Note that $\mathbf{y}^h_\mathrm{in} \in \partial Q_{i_h, j_h}$ for $h>1$.  
Similarly, let $\mathbf{y}^h_\mathrm{out} \in \bar{Q}_{i_h, j_h}$ be the intersection point that is the closest to $\mathbf{x}^h_\mathrm{out}$. 
As before $\mathbf{y}^h_\mathrm{out} \in \partial Q_{i_h, j_h}$ for $h < k$. 
Consider the path $s, v_{\mathbf{y}^1_\mathrm{out}}, v_{\mathbf{y}^2_\mathrm{in}}, v_{\mathbf{y}^2_\mathrm{out}}, \ldots, v_{\mathbf{y}^{k-1}_\mathrm{in}}, v_{\mathbf{y}^{k-1}_\mathrm{out}}, v_{\mathbf{y}^{k}_\mathrm{in}}, d$.
The theorem below bounds the length of this path. 
For two vectors $\mathbf{x}, \mathbf{y}$, let $\theta(\mathbf{x}, \mathbf{y}) \in [0, \pi]$ be the angle between $\mathbf{x}$ and $\mathbf{y}$. 

\begin{theorem}
\label{thm:upper bound}
For a demand $(s,d) \in \mathcal{D}$ and $\rho $ sufficiently large as described above, the length of the shortest path from $s$ to $d$ on our topology can be bounded as 
\begin{align}
l(p_{s,d}) \leq \left(1 + \frac{4}{R}(2\eta + 1/(\sqrt{2}\rho)) \right) \left( (R + \| \mathbf{x}_s - \mathbf{y}^1_\mathrm{out} \|(|\cos(\theta(\mathbf{x}_s - \mathbf{y}^1_\mathrm{out}, \mathbf{p}(R(i_1, j_1))))| \right. \notag \\
 + |\sin(\theta(\mathbf{x}_s - \mathbf{y}^1_\mathrm{out}, \mathbf{p}(R(i_1, j_1))))| ) ) \notag \\
+ \sum_{h=2}^{k-1} (R + \| \mathbf{y}^h_\mathrm{in} - \mathbf{y}^h_\mathrm{out} \|(|\cos(\theta(\mathbf{y}^h_\mathrm{in} - \mathbf{y}^h_\mathrm{out}, \mathbf{p}(R(i_h, j_h))))| + |\sin(\theta(\mathbf{y}^h_\mathrm{in} - \mathbf{y}^h_\mathrm{out}, \mathbf{p}(R(i_h, j_h))))| ) )\notag \\
\left. + (R + \| \mathbf{x}_d - \mathbf{y}^k_\mathrm{in} \|(|\cos(\theta(\mathbf{x}_d - \mathbf{y}^k_\mathrm{in}, \mathbf{p}(R(i_k, j_k))))| + |\sin(\theta(\mathbf{x}_d - \mathbf{y}^k_\mathrm{in}, \mathbf{p}(R(i_k, j_k))))| ) ) \right) + (k-1)\frac{R}{2}.  
\end{align} 
\end{theorem}
\begin{proof}
Consider the ideal grid formed by the points $\bar{Q}_{i_1,j_1}$ oriented along the direction $\mathbf{p}(R_{i_1,j_1})$.  
Consider the shortest path from $\mathbf{x}_s$ to $\mathbf{y}^1_\mathrm{out}$ on the ideal grid. 
The length of this path is $\| \mathbf{x}_s - \mathbf{y}^1_\mathrm{out} \| ( | \cos(\theta(\mathbf{p}(R_{i_1,j_1}), \mathbf{x}_s - \mathbf{y}^1_\mathrm{out})) | + | \sin(\theta(\mathbf{p}(R_{i_1,j_1}), \mathbf{x}_s - \mathbf{y}^1_\mathrm{out}))| )$, and comprises of at most $2 + \frac{2}{R} \| \mathbf{x}_s - \mathbf{y}^1_\mathrm{out} \| ( | \cos(\theta(\mathbf{p}(R_{i_1,j_1}), \mathbf{x}_s - \mathbf{y}^1_\mathrm{out})) | + | \sin(\theta(\mathbf{p}(R_{i_1,j_1}), \mathbf{x}_s - \mathbf{y}^1_\mathrm{out}))| )$ hops. 
Now, consider the same path in the actual topology where each point in $\bar{Q}_{i,j}$ is replaced with the satellite closest to it. 
While the edge length in the ideal grid is exactly $R/2$ in the actual topology the edge length is at most $R/2 + 2(2\eta + 1/(\sqrt{2}\rho))$. 
Hence, an upper bound for the path length from $s$ to $v_{\mathbf{y}^1_\mathrm{out}}$ in the actual topology is 
\begin{align}
(R/2 + 2(2\eta + 1/(\sqrt{2}\rho)))( 2 + \frac{2}{R} \| \mathbf{x}_s - \mathbf{y}^1_\mathrm{out} \| ( | \cos(\theta(\mathbf{p}(R_{i_1,j_1}), \mathbf{x}_s - \mathbf{y}^1_\mathrm{out})) | \notag \\
+ | \sin(\theta(\mathbf{p}(R_{i_1,j_1}), \mathbf{x}_s - \mathbf{y}^1_\mathrm{out}))| ) ). 
\end{align}
Next, to go from $v_{\mathbf{y}^1_\mathrm{out}}$ to $v_{\mathbf{y}^2_\mathrm{in}}$ the cost is at most $R/2$. 
By continuing this process for all the regions, the theorem follows.
\end{proof}

If the demand $(s,d) \in \mathcal{D}$ does not have other competing demands in its vicinity, the primary direction $\mathbf{p}(R(i_h, j_h))$ is parallel to $\mathbf{x}_d - \mathbf{x}_s$ for $h=1,2,\ldots,k$. 
Also, for any $h \in \{1,2,\ldots,k\}$, the points $\mathbf{y}^h_\mathrm{in}, \mathbf{y}^h_\mathrm{out}$ are at a distance of at most $\frac{R}{2\sqrt{2}}$ to $\mathbf{x}^h_\mathrm{in}, \mathbf{x}^h_\mathrm{out}$ respectively. 
Therefore, the angle between $\mathbf{y}^h_\mathrm{in} - \mathbf{y}^h_\mathrm{out}$ and  $\mathbf{x}^h_\mathrm{in} - \mathbf{x}^h_\mathrm{out}$ differs by at most $\sin^{-1}\left(\frac{R}{2\sqrt{2}\| \mathbf{y}^h_\mathrm{in} - \mathbf{y}^h_\mathrm{out} \|} \right)$. 
We therefore have the following corollary. 
\begin{corollary}
If the demand $(s,d) \in \mathcal{D}$ does not have other competing demands in its vicinity with $\mathbf{p}(R(i_h, j_h))$ parallel to $\mathbf{x}_d - \mathbf{x}_s$ for $h=1,2,\ldots,k$, the length of the shortest path from $s$ to $d$ can be bounded as
\begin{align}
l(p_{s,d}) \leq \left(1 + \frac{4}{R}(2\eta + 1/(\sqrt{2}\rho)) \right) \left( (R + \| \mathbf{x}_s - \mathbf{y}^1_\mathrm{out} \|(1  + \frac{R}{2\sqrt{2}\| \mathbf{x}_s - \mathbf{y}^1_\mathrm{out} \|}  ) ) \right. \notag \\
+ \sum_{h=2}^{k-1} (R + \| \mathbf{y}^h_\mathrm{in} - \mathbf{y}^h_\mathrm{out} \|(1 + \frac{R}{2\sqrt{2}\| \mathbf{y}^h_\mathrm{in} - \mathbf{y}^h_\mathrm{out} \|}  ) )\notag \\
\left. + (R + \| \mathbf{x}_d - \mathbf{y}^k_\mathrm{in} \|(1 + \frac{R}{2\sqrt{2}\| \mathbf{x}_d - \mathbf{y}^k_\mathrm{in} \|} ) ) \right) + (k-1)\frac{R}{2}.  
\end{align} 
\end{corollary}

\end{document}